\newcommand*{\FINAL}{} %%FINAL reomve comments or REVISION to show author comments

\documentclass[sigplan]{acmart}

\acmConference[]{}{}{}
\acmYear{2023}
\acmISBN{} % \acmISBN{978-x-xxxx-xxxx-x/YY/MM}
\acmDOI{} % \acmDOI{10.1145/nnnnnnn.nnnnnnn}
\startPage{1}

\setcopyright{rightsretained}
\copyrightyear{2023}           %% If different from \acmYear

\usepackage{algorithm,algorithmic}

\usepackage{tikz}
\usepackage{bbm}

\usepackage{aliascnt}
\usepackage{multirow}
\usepackage{subcaption}
\usepackage{xspace}
\usepackage{wrapfig}
\usepackage{placeins}
\usepackage{titlesec}
\usepackage{lipsum}
\usepackage{setspace}

\tikzset{black node/.style={draw, circle, fill = black, minimum size = 5pt, inner sep = 0pt}}
\tikzset{normal/.style = {draw=none, fill = none, minimum size =0, rectangle}}

\newcommand{\spmv}{SpMV\xspace}
\DeclareMathOperator*{\argmin}{arg\,min}

\newaliascnt{prop}{theorem}
\newtheorem{prop}[prop]{Proposition}
\aliascntresetthe{prop}

\newaliascnt{lemma}{theorem}
\newtheorem{lemma}[lemma]{Lemma}
\aliascntresetthe{lemma}

\newaliascnt{observation}{theorem}

\aliascntresetthe{observation}

\newaliascnt{model}{theorem}
\newtheorem{model}[model]{Model}
\aliascntresetthe{model}

\newaliascnt{corollary}{theorem}
\newtheorem{corollary}[corollary]{Corollary}
\aliascntresetthe{corollary}

\newaliascnt{conjecture}{theorem}

\aliascntresetthe{conjecture}

\newaliascnt{claim}{theorem}

\aliascntresetthe{claim}

\newaliascnt{problem}{theorem}
\newtheorem{problem}[problem]{Problem}
\aliascntresetthe{problem}

\newaliascnt{question}{theorem}
\newtheorem{question}[question]{Question}
\aliascntresetthe{question}

\theoremstyle{definition}

\newcommand{\defn}[1]{\textcolor{blue}{\emph{#1}}\xspace}

\newcommand{\secref}[1]{Section~\ref{#1}}
\bibpunct{\nolinebreak{}[}{]}{,}{n}{}{,}

\ifdefined\FINAL
\else
    \definecolor{darkgreen}{RGB}{0,100,0}
\fi

\begin{document}
%\title{A Light Weight Graph Reordering Algorithm for Edge Lists}
%\title{BOBA: A Parallel Lightweight Graph Reordering Algorithm with Heavyweight (Impact, Implications, Results, Performance, Speed)}
%BOBA: A Lightweight BLAH that delivers Heavyweight results

\title{BOBA\@: A Parallel Lightweight Graph Reordering Algorithm with Heavyweight Implications}

%% Author information
%% Contents and number of authors suppressed with 'anonymous'.
%% Each author should be introduced by \author, followed by
%% \authornote (optional), \orcid (optional), \affiliation, and
%% \email.
%% An author may have multiple affiliations and/or emails; repeat the
%% appropriate command.
%% Many elements are not rendered, but should be provided for metadata
%% extraction tools.

%% Author with single affiliation.
\author{Matthew Drescher}
\authornote{This material is based upon work supported by the National Science Foundation under Grant \#2127309 to the Computing Research Association for the CIFellows 2021 Project. Any opinions,findings, and conclusions or recommendations expressed in this material are those of the authors and do not necessarily reflect the views of the National Science Foundation or the Computing Research Association.}
%% can be repeated if necessary
%\orcid{nnnn-nnnn-nnnn-nnnn}             %% \orcid is optional
\affiliation{
  %\position{Position1}
  %\department{Department1}              %% \department is recommended
  \institution{University of California, Davis}            %% \institution is required
  %\streetaddress{Street1 Address1}
  \city{Davis}
  \state{California}
  %\postcode{Post-Code1}
  \country{U.S.A}                    %% \country is recommended
}
\email{mdrescher@ucdavis.edu}          %% \email is recommended

\author{Muhammad A. Awad}
% \authornote{with author1 note}          %% \authornote is optional;
%% can be repeated if necessary
\orcid{0000-0002-6914-493X}             %% \orcid is optional
\affiliation{
  %\position{Position1}
  %\department{Department1}              %% \department is recommended
  \institution{University of California, Davis}            %% \institution is required
  %\streetaddress{Street1 Address1}
  \city{Davis}
  \state{California}
  %\postcode{Post-Code1}
  \country{U.S.A}                    %% \country is recommended
}
\email{mawad@ucdavis.edu}          %% \email is recommended

\author{Serban D. Porumbescu}
% \authornote{with author1 note}          %% \authornote is optional;
%% can be repeated if necessary
\orcid{0000-0003-1523-9199}             %% \orcid is optional
\affiliation{
  % \position{Position1}
  % \department{Department1}              %% \department is recommended
  \institution{University of California, Davis}            %% \institution is required
  % \streetaddress{Street1 Address1}
  \city{Davis}
  \state{California}
  % \postcode{Post-Code1}
  \country{U.S.A}                    %% \country is recommended
}
\email{sdporumbescu@ucdavis.edu}          %% \email is recommended

\author{John D. Owens}
% \authornote{with author1 note}          %% \authornote is optional;
%% can be repeated if necessary
\orcid{0000-0001-6582-8237}            %% \orcid is optional
\affiliation{
  % \position{Position1}
  %\department{Department1}              %% \department is recommended
  \institution{University of California, Davis}            %% \institution is required
  %\streetaddress{Street1 Address1}
  \city{Davis}
  \state{California}
  %\postcode{Post-Code1}
  \country{U.S.A}                    %% \country is recommended
}
\email{jowens@ece.ucdavis.edu} %% \email is recommended

\begin{abstract}
  We describe a simple parallel-friendly lightweight graph reordering algorithm for COO graphs (edge lists). Our
  ``Batched Order By Attachment'' (BOBA) algorithm is linear in the number of edges in terms of reads and linear in the number of vertices for writes through to main memory. It is highly parallelizable on GPUs\@. We show that, compared to a randomized baseline, the ordering produced gives improved locality of reference in sparse matrix-vector multiplication (SpMV) as well as other graph algorithms. Moreover, it can substantially speed up the conversion from a COO representation to the compressed format CSR, a very common workflow. Thus, it can give \emph{end-to-end} speedups even in SpMV\@. Unlike other lightweight approaches, this reordering does not rely on explicitly knowing the degrees of the vertices, and indeed its runtime is comparable to that of computing degrees. Instead, it uses the structure and edge distribution inherent in the input edge list, making it a candidate for default use in a pragmatic graph creation pipeline. This algorithm is suitable for road-type networks as well as scale-free. It improves cache locality on both CPUs and GPUs, achieving hit rates similar to the heavyweight techniques (e.g., for SpMV, 7--52\% and 11--67\% in the L1 and L2 caches, respectively). Compared to randomly labeled graphs, BOBA-reordered graphs achieve end-to-end speedups of up to 3.45. The reordering time is approximately one order of magnitude faster than existing lightweight techniques and up to 2.5 orders of magnitude faster than heavyweight techniques.
\end{abstract}

%% 2012 ACM Computing Classification System (CSS) concepts
%% Generate at 'http://dl.acm.org/ccs/ccs.cfm'.

%% End of generated code

%% Keywords
%% comma separated list
  %% \keywords are mandatory in final camera-ready submission

%% \maketitle
%% Note: \maketitle command must come after title commands, author
%% commands, abstract environment, Computing Classification System
%% environment and commands, and keywords command.
\maketitle

\section{Introduction}
\label{sec:intro}

Graph data structures typically encode each of the $n$ vertices of a graph with a unique ID from $1$ to $n$. Edges, either explicitly or implicitly, are then encoded as pairs of vertices, and the entire graph data structure is then stored into a block of memory in vertex-ID order. Because the actual value of a vertex's ID is usually unimportant, we have a great deal of freedom to optimize the ordering of vertices in a graph data structure in the service of a particular goal.

Now, graph computations often spend much if not most of their time traversing edges in the graph from source to destination. Because graphs typically exhibit complex connectivity, and the size of an interesting graph is usually much larger than the size of any last-level cache, a random ordering of vertices is unlikely to achieve significant cache locality. As a result, many graph computations are dominated by random access into main memory.

Could we reorder the vertices in the graph to recover locality? Previous work in this area has shown that reordering can successfully increase performance by exposing locality in graph computation. Such reordering efforts---for instance, reducing the bandwidth of the graph, which places connected vertices near each other, or partitioning the graph to expose locality within partitions---have achieved significant speedups but are expensive, with the reordering process taking much more time than the subsequent graph computation. Consequently, these heavyweight methods are primarily useful in \emph{offline} scenarios where a graph is pre-processed once but used many times, so that the cost of reordering can be amortized across many uses.

Our work addresses a different use scenario. We focus on a \emph{lightweight} reordering that is inexpensive to compute and thus can be useful even in \emph{online} scenarios where we have no opportunity to preprocess the graph.  Such scenarios are common in modern data-science workflows like NVIDIA's RAPIDS, where graph computation may be an intermediate stage of a complex pipeline that produces graph data dynamically and where preprocessing is not an option. The ideal reordering process achieves the performance of heavyweight (offline) methods while remaining inexpensive enough to demonstrate performance benefits even for graphs where pre-processing is impossible. In other words, the ideal technique would achieve better performance for the combination of reordering and graph computation than the performance of graph computation alone on the un-reordered graph.
\FloatBarrier\subsection{Defining the Problems}
\label{sub:defining}
\begin{problem}[Offline Graph Reordering]\label{prob:offline}
Given a graph $G$ and a graph application $f(G)$, find an ordering of $G$'s vertices, in time polynomial in $V(G)$, to maximize cache locality, with the expectation that better cache locality maximizes performance of $f(G)$.
\end{problem}
We refer to methods that target Problem~\ref{prob:offline} as \defn{heavyweight} algorithms.
From a more pragmatic point of view, the reason for reordering a graph is to accelerate some graph application. Keeping this context in mind motivates:

\begin{problem}[Online Graph Reordering]\label{prob:online}
Given a graph $G$ and a graph application $f(G)$, find an ordering of $G$'s vertices, that can, even including the cost of reordering, improve cache-locality such that there is a net speedup in $f(G)$ .
\end{problem}

Similarly we refer to methods that target Problem~\ref{prob:online} as \defn{lightweight} algorithms.
The most common starting point for \emph{building} a matrix is a COO representation~\cite{Filippone:2017:SMM}. This representation follows naturally from most file formats, where an edge-list representation is common if not dominant.\footnote{For instance, SuiteSparse (\url{https://sparse.tamu.edu/}), networkrepository (\url{https://networkrepository.com/}), and Stanford SNAP (\url{https://snap.stanford.edu/}) primarily use \texttt{el} and/or \texttt{mtx} edge-list formats.} In the following discussion, we use \spmv (single-hop graph traversal from all graph vertices) to represent any graph computation; \spmv is both an important kernel as well as a simple one, so if we can satisfy Problem~\ref{prob:online} with \spmv, we can reasonably expect similar success with other graph kernels.

While some implementations of \spmv run directly on a COO representation, more common is first \emph{converting} an edge-list representation to a CSR, the most popular format for computation~\cite{Filippone:2017:SMM}. This resulting CSR representation is typically presented as the input for graph reordering and \spmv. This is often a convenient assumption, since in the conversion to CSR, vertex degree has essentially been pre-computed.

Indeed, popular real world frameworks for data-science such as SciPy\footnote{\url{https://docs.scipy.org/doc/scipy/reference/generated/scipy.sparse.coo_matrix.html} SciPy's supported function for reading a Matrix Market file, \texttt{mmread} returns \emph{only} COO format},  NetworkX\footnote{\url{https://networkx.org/documentation/stable/reference/readwrite/matrix_market.html}, COO is also \emph{the} supported path for reading Matrix Market files}, RAPIDS\footnote{When RAPIDS reads a Matrix Market graph file, it first creates an edge list (COO). \url{https://github.com/rapidsai/cugraph/blob/7d8f0fd63ad58ce6deada5508bfc08ee9aa46d36/cpp/tests/utilities/matrix_market_file_utilities.cu}}, as well as GPU graph frameworks such as Gunrock~\cite{Wang:2017:GGG}\footnote{\url{https://github.com/gunrock/gunrock/blob/a7fc6948f397912ca0c8f1a8ccf27d1e9677f98f/gunrock/graphio/market.cuh}}, follow this process, with the additional complication that vertices are often not numerically labeled. In such workflows, relabeling vertices to numeric IDs is already necessary, and since BOBA does not require its input edge list to have numeric IDs, but returns a cache-friendly numeric ordering, BOBA is a natural fit\footnote{\url{https://github.com/rapidsai/cugraph/blob/492245009cd2075054573b450a602422ae8f4a78/python/cugraph/cugraph/tests/test_renumber.py}}.
Therefore we motivate the primary problem:

\begin{problem}[Pragmatic Graph Reordering]
Given a COO representation of a graph with $n$ randomly labeled vertices from the set $\{1,2,\ldots,n\}$, is there a reordering algorithm, that can, even including the cost of reordering, give a net speedup in CSR graph creation and \spmv on the resultant CSR?
\label{problem:pragmatic-graph-reordering}
\end{problem}
We find that BOBA answers the question in the affirmative, and moreover, perhaps surprisingly, that its results are competitive with existing heavyweight reorderings.  Therefore we motivate the following questions with respect to the \emph{pragmatic graph reordering problem}.

It is important to clarify that as an abstract graph algorithm, BOBA can easily be implemented on any graph representation from which we can extract an edge list. We chose COO for this paper, as we think it is at this stage of the graph construction pipeline that BOBA is uniquely effective.

Reordering prior to COO$\rightarrow$CSR conversion speeds up that conversion considerably. In Section~\ref{sec:results}, we show a significant speedup in conversion time. The intuition behind this speedup is that BOBA tends to improve spatial locality of the neighborhoods of vertices. Thus when traversing the COO's edge list to create a CSR, we incur fewer cache misses. BOBA is profitable for the gains given in this conversion alone.

\begin{question}[Offline]
    How does BOBA compare to other reordering methods as an offline reordering method?
\end{question}

\begin{question}[Online]
    How does BOBA compare to other reordering methods as an online reordering method?
\end{question}

\begin{figure}[t]\centering
    \captionsetup{font=footnotesize,labelfont=footnotesize}
    \begin{minipage}[t]{.32\textwidth}
        \begin{tikzpicture}[scale=.3,inner sep=1.5pt]
            \tikzstyle{vtx}=[circle,draw,thick,fill=white]
            \node[vtx, fill = orange] (a) at (4,0) {\tiny{$a$}};
            \node[vtx] (v1) at (4,4) {\tiny{1}};
            \node[vtx] (v5) at (4,-4) {\tiny{5}};
            \node[vtx] (v3) at (0,0) {\tiny{3}};
            \node[vtx] (v2) at (1,2.5) {\tiny{2}};
            \node[vtx] (v4) at (1,-2.5) {\tiny{4}};

            \node[vtx, fill = orange] (b) at (8,0) {\tiny{$b$}};
            \node[vtx] (v6) at (8,4) {\tiny{6}};
            \node[vtx] (v10) at (8,-4) {\tiny{10}};
            \node[vtx] (v8) at (12,0) {\tiny{8}};
            \node[vtx] (v7) at (11,2.5) {\tiny{7}};
            \node[vtx] (v9) at (11,-2.5) {\tiny{9}};

            \draw[thick,->] (v3)--(a); %--(v1)--(a)--(v5)--(a)--(v3)--(a)--(b);
            \draw[thick,->] (v1)--(a);
            \draw[thick,->] (v2)--(a);
            \draw[thick,->] (v5)--(a);
            \draw[thick,->] (v4)--(a);
            \draw[thick,->] (b)--(a);
            % \draw[thin] (4,0)--(1,2.5)--(4,0)--(4,0)--(1,-2.5);
            \draw[thick,->] (v6)--(b); %--(v1)--(a)--(v5)--(a)--(v3)--(a)--(b);
            \draw[thick,->] (v7)--(b);
            \draw[thick,->] (v8)--(b);
            \draw[thick,->] (v9)--(b);
            \draw[thick,->] (v10)--(b);
            \draw[thick,->] (a)--(b);

        \end{tikzpicture}
    \end{minipage}

    \begin{minipage}[b]{.5\textwidth}
        \begin{tikzpicture}[scale=.6,inner sep=1pt, every node/.style={anchor=west,font=\small}]
            \matrix [nodes={draw,fill=blue!3,minimum size=2.5mm}, anchor=south] {
                \node{\tiny 1};   & \node{\tiny 2};   & \node{\tiny 3};   & \node{\tiny 4};
                                  & \node{\tiny $b$}; & \node{\tiny 5};
                                  & \node{\tiny 6};   & \node{\tiny 7};   & \node{\tiny 8};   & \node{\tiny 9};
                                  & \node{\tiny 10};  & \node{\tiny $a$};                                                                                 \\
                \node{\tiny $a$}; & \node{\tiny $a$}; & \node{\tiny $a$}; & \node{\tiny $a$}; & \node{\tiny $a$}; & \node{\tiny $a$};
                                  & \node{\tiny $b$}; & \node{\tiny $b$}; & \node{\tiny $b$}; & \node{\tiny $b$}; & \node{\tiny $b$}; & \node{\tiny $b$}; \\
            };
            \draw (3,.8) node{I};
            \draw (3,.2) node{J};
        \end{tikzpicture}

        \begin{tikzpicture}[scale=.6,inner sep=1pt,every node/.style={anchor=west,font=\small}]
            \matrix [nodes={draw,fill=green!3,minimum size=2.5mm}] {
                \node{\tiny 1}; & \node{\tiny $a$}; & \node{\tiny 2};   & \node{\tiny $a$};
                                & \node{\tiny 3};   & \node{\tiny $a$}; & \node{\tiny 4};   & \node{\tiny $a$};
                                & \node{\tiny $b$}; & \node{\tiny $a$}; & \node{\tiny 5};   & \node{\tiny $a$};
                                & \node{\tiny 6};   & \node{\tiny $b$}; & \node{\tiny 7};   & \node{\tiny $b$};
                                & \node{\tiny 8};   & \node{\tiny $b$}; & \node{\tiny 9};   & \node{\tiny $b$};
                                & \node{\tiny 10};  & \node{\tiny $b$}; & \node{\tiny $a$}; & \node{\tiny $b$}; \\
            };
            \draw (8,.5) node{\tiny{Flattened Edge List}};
        \end{tikzpicture}

        \begin{tikzpicture}[scale=.6,inner sep=1pt,every node/.style={anchor=west,font=\small}]
            \matrix [nodes={draw,fill=yellow!3,minimum size=2.5mm}] {
                \node{\tiny 1}; & \node{\tiny $a$}; & \node{\tiny 2};
                                & \node{\tiny 3};   & \node{\tiny 4}; & \node{\tiny $b$}; & \node{\tiny 5}; & \node{\tiny 6}; & \node{\tiny 7};
                                & \node{\tiny 8};   & \node{\tiny 9}; & \node{\tiny 10};                                                        \\
            };
            \draw (6,0) node{\tiny{Uniquified}};
        \end{tikzpicture}
    \end{minipage}

    \caption{\textbf{Top} A star-like graph with two centers $a$, and $b$ that are adjacent. A uniformly random permutation of $\{1,2,\ldots,10,a,b\}$ would be as likely to map $a,b$ far apart as close together. Instead, BOBA uniformly selects a cell of the flattened edge list, adds its vertex to the permutation, and then removes all remaining cells containing it, bringing $a,b$ closer together. Such a selection is similar to the attachment model of Albert and Barab{\'a}si~\cite{albert2002statistical}.
        Let $p_k$ be the probability that $a,b$ will both appear in the first $k$ positions. $p_2 = \frac{14}{24}\cdot\frac{7}{17} \approx 24\% $, $p_3 \approx 50\%, p_4 \approx 70\%$. Thus both will most likely occur within the first $\sim$5 positions.
    }
    \label{fig:star}

\end{figure}
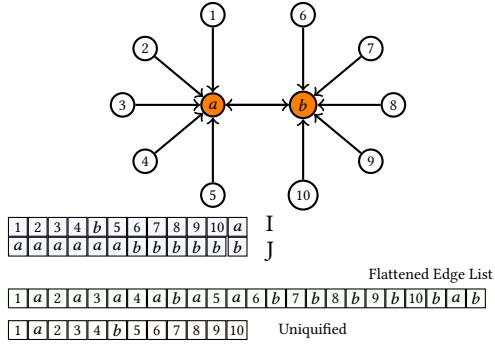

\subsection{Our Contribution}
This paper introduces \textsc{Batched Order By Attachment} (BOBA) a fast method for reordering graph data to take better advantage of hardware locality. It is inspired by \defn{preferential attachment}(PA), a network generation process defined by Albert and Barab{\'a}si~\cite{albert2002statistical}, which is used to mathematically model the structure of real-world scale-free networks~\cite{newman2018networks}. PA is a process that iteratively generates a scale-free graph. In PA, the $ith$ vertex $v_i$ \emph{attaches} to existing vertices $j \in \{1,2,\ldots,i-1\}$, with probability proportional to the current degree of $v_j$. In other words, at time $i$, $v_i$ will \emph{prefer} to connect with vertices that \emph{appear the most often in the edge list}. This can be accomplished by randomly selecting vertices from a flattened version of the edge list. Note, this does not require the calculation of degree information. The idea of the BOBA algorithm is similar, though we are given a graph $G$, not generating it: \emph{what if we run the PA process starting with $G$ as an initial state, and then let the new vertices decide the order of $G$ by the time they are selected for attachment?} This can be crudely approximated simply, albeit with some ambiguity: \emph{order vertices by their appearance in the edge list}. We note \emph{any} graph representation allows for a fairly straightforward implementation of this idea. The COO format makes it especially easy: Create a permutation of the vertices by concatenating the sequence of \emph{sources} with the sequence of \emph{destinations} and then \emph{uniquify the resulting sequence in a stable order}. The resulting permutation decreases the execution time of graph applications because of the following.
\subsubsection{Temporal Locality}
According to Esfahani et al.~\cite{esfahani:2021:locality}, Gorder~\cite{Wei:2016:SGP} primarily targets \emph{temporal locality}. That is, these orderings try to increase the chance that when a vertex's neighbors are brought into cache, the next vertex processed will read some of these same cached neighbors. This is also the principle behind degree-based ordering approaches. In scale-free graphs, so called `hub' (high-degree) vertices can be moved together and placed into the same cache line; the vertices in this cache line are internally densely connected and likely adjacent to $\sim$most vertices in the graph. However, degree-based sorting approaches are inappropriate for other types of degree distributions, where, in fact, they make things worse (e.g., Figure~\ref{fig:road}).

Figure~\ref{fig:star} is a simple example showing that in scale-free distributions, BOBA tends to bring higher-degree vertices closer together, thus obtaining some of the same type of temporal locality enjoyed by degree orderings. Moreover, BOBA tends to also produce better results in more uniform types of distributions. We show in Proposition~\ref{prop:regular} that in a uniform setting, BOBA gives a constant-factor approximation guarantee for the neighborhood problem that Gorder explicitly targets. BOBA, on the other hand, more deliberately targets spatial locality.

% In sparse, scale-free networks, a small number of `hub' vertices have very high degree, whereas most other vertices are of low degree. Since hubs are connected to a large number of vertices, they are often connected to each other and tend to form a densely connected subgraph.  Existing lightweight methods aim to order this subgraph contiguously through some means of sorting or partially sorting by previously computed degrees. BOBA is less deliberate. Observe that a random permutation selects a vertex $v$ with probability of $\frac{1}{n}$ regardless of its degree. On the other hand, selecting a vertex from a flattened edge list, as BOBA does, $v$ is selected with probability $\frac{\deg{v}}{Cn}$ where $C$ is the average degree. Thus, hub vertices, whose degrees are much greater than average, are selected with considerably higher probability and should occur proportionately earlier than in a random permutation, and likewise, vertices with below-average degree occur later.
\subsubsection{Spatial Locality}
Since BOBA operates directly on the edge listing, it tends to cluster source vertices that share a common destination vertex in the output ordering. This aids spatial locality, and in directed graphs, helps pull-based algorithms. For example, below we list a pull-based \spmv, $y = Ax$, where $x$ is an input array of size $|V(G)|$, and $N^\text{in}(v)$ denotes the \defn{in-neighbors} of $v$.

\begin{algorithm}\caption{\sc{\spmv: Pull}}\label{alg:graph}
    \begin{scriptsize}
        \begin{algorithmic}[1] \label{code:neighborhood}
            \REQUIRE a graph $G$ with $n$ vertices, $x \in \mathbb{R}^n$, $A_G :=$ adjacency matrix of $G$.
            \ENSURE $y = A_G \cdot x$.
            \STATE initialize $y \gets \mathbf{0}^n$.
            \FORALL{$v \in V(G)$ in parallel}
            \FOR{$ u \in N^\text{in}(v)$}
            \STATE{$y(v) \gets y(v) + A_G(v,u) \cdot x(u)$} \label{line:rand}
            \ENDFOR
            \ENDFOR
        \end{algorithmic}
    \end{scriptsize}
\end{algorithm}

Here, in this pull-based pattern, the inner loop's performance is affected by the cache performance of accesses into $N^\text{in}(v)$. A CSC (resp.\ CSR) representation offers cache friendly access to $N^\text{in}(v)$ (resp.\ $N^\text{out}(v)$). However, random accesses into dense vector $x$ (Line~\ref{line:rand}) result in poor cache performance. Reordering attempts to minimize this by moving the non-zero columns of A (vertices of G) closer together. To give some measure of the quality of this \emph{spatial locality}, we develop the metric \emph{per neighborhood cache-line}. This scores the number of cache-lines spanned by the IDs of each neighborhood, where lower scores imply better spatial locality. BOBA scores well on this metric; in Section~\ref{sec:results} we show that BOBA's ordering is significantly better than random, and that the metric correlates with performance.

\FloatBarrier\subsubsection{BOBA does not harm natural structure}
For graphs with $\sim$uniform degree or where degree is anti-correlated with connectivity (Figure~\ref{fig:road}), degree-based sorting algorithms are no better than random. Other heavyweight methods destroy the inherent community structure that is sometimes present in the original ordering, or are simply ineffective.\footnote{Gorder, for instance, despite high run-time cost, does not give significant speedups compared to degree based sort methods on graphs, e.g., \texttt{kron\_g500-logn20}, that have very low average clustering coefficients~\cite{faldu2019closer}.} In contrast, BOBA actually seems to restore this original structure in graphs whose original generation process is similar to Preferential Attachment (e.g., Figure~\ref{fig:pa}). We point out that as a simple corollary of Bollobas's Theorem 16~\cite{Bollobas:2002:MRO}\label{lem:bollobas}, order by attachment time, in graphs that follow this model, scores well by the shared-neighbors metric.

\begin{figure}
    %\begin{wrapfigure}{r}{0.3\textwidth}
    \centering
    \captionsetup{font=footnotesize,labelfont=footnotesize}
    \begin{subfigure}[t]{\columnwidth}
        \includegraphics[width=\columnwidth]{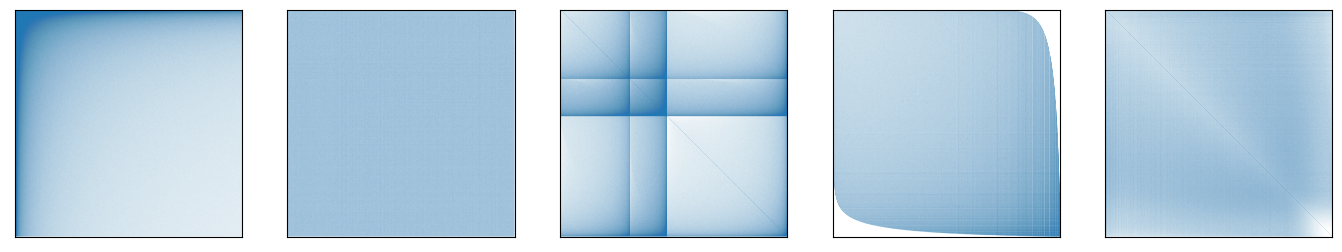}
        \caption{Barab\'{a}si-Albert graph dataset generated with 5,000,000 nodes and 10 edges (attachments) per node insertion. }
        \label{subfig:barabasi}
    \end{subfigure}

    \begin{subfigure}[t]{\columnwidth}
        \includegraphics[width=\columnwidth]{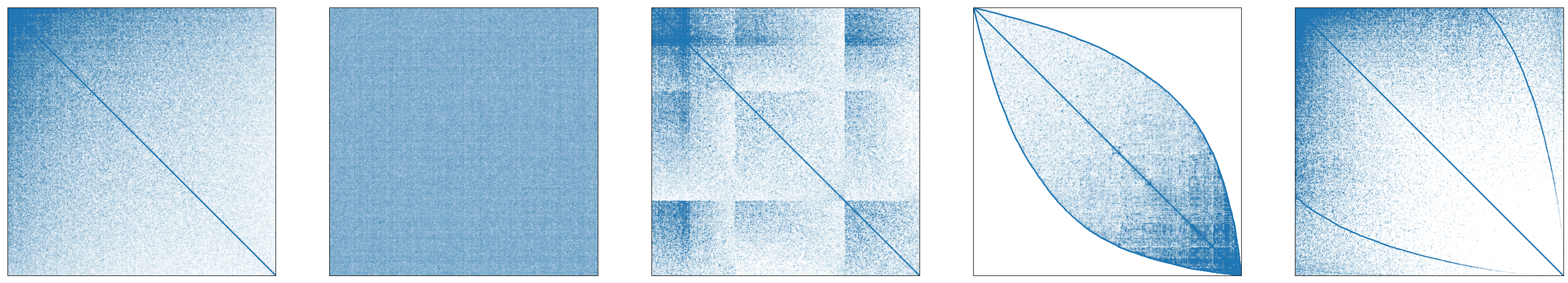}
        \caption{The coAuthors dataset. }
        \label{subfig:coauthors}
    \end{subfigure}
    \begin{subfigure}[t]{\columnwidth}
        \includegraphics[width=\columnwidth]{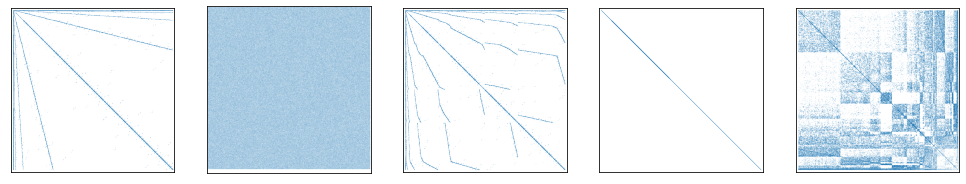}
        \caption{The delaunay\_n22 dataset.}
        \label{subfig:delaunay}
    \end{subfigure}

    \caption{Visualizations of (\ref{subfig:barabasi}) a simulated powerlaw graph, (\ref{subfig:coauthors}) a real-world powerlaw graph, and (\ref{subfig:delaunay}) a regular uniform graph under five different reorderings, from left to right: original dataset, randomized ordering, BOBA (middle), RCM, Gorder.
        Column 2 (from the left) displays the randomized ordering used as input for BOBA, RCM, and Gorder. We see that the BOBA ordering captures more of the spatial structures seen in the original, unordered dataset. This is most easily seen in (\ref{subfig:delaunay}).
    }
    \label{fig:pa}
\end{figure}
%\end{wrapfigure}

The case where BOBA alone can not help is a uniform degree distribution where the input edge list appears in random order. We discuss options for this case in Section~\ref{subsec:randomEdges}.  Other lightweight approaches that rely on scale-free degree distributions also cannot help in this case, and can actually make things worse. Indeed, when vertex degree is roughly uniform, sorting by degree is roughly a random ordering (Figure~\ref{fig:road}). Even when BOBA cannot improve cache utilization, its reordering cost is minimal, BOBA is safe from a reorder time investment as well.

\subsubsection{BOBA is fast}
When heavyweight methods don't succeed in delivering dramatic speedups, it can be a considerable loss of investment, whereas for BOBA, it is quite minimal. For example, BOBA took only 17~ms to reorder the $\sim$90 million-edge \texttt{kron\_g500-logn20}, while Gorder took 42~s.

BOBA is faster than existing lightweight methods, especially when degree information is unknown (e.g., when the input is COO)\@. This makes it uniquely suited for COO graphs. BOBA can give significant gains to both CSR conversion \emph{and} applications like \spmv. These gains easily make up for the small reordering cost. For example, on a modern GPU, BOBA can reorder a graph with roughly 60 million edges in roughly 16 milliseconds. This might reduce the runtime of \spmv from 9 to 5 milliseconds, for a net loss of 11 milliseconds, but when we factor in that COO~$\rightarrow$ CSR conversion also decreased, from 8000 milliseconds to 5000, we see that the gains are substantial. Though some other reordering techniques (e.g., RCM) are challenging to parallelize~\cite{Karantasis:2014:PRA}, BOBA is trivially parallelizable, and in fact the benchmarks we present here have been run via parallel implementation on a GPU, whose programming model is a good match for our algorithm. We suggest that BOBA should be applied indiscriminately to unordered, or randomly labeled, graph data.

\begin{figure}[t]\centering
    \captionsetup{font=footnotesize,labelfont=footnotesize}
    \begin{minipage}[t]{.2\textwidth}
        \begin{tikzpicture}[scale=.3,inner sep=1pt]
            \tikzstyle{vtx}=[circle,draw,thick,fill=white]
            \draw[thick] (-2.5,-2)--(0,0)--(2,0)--(4,0)--(4,2)--(4,4)--(2,4)--(0,4);
            \draw[thick] (-1,-2)--(0,0);
            \draw[thick] (.5,-2.5)--(0,0);
            \draw[thick] (0,4)--(0,2)--(0,0);
            \draw[thick] (6,6)--(4,4)--(6,4);
            \draw (0,0) node[vtx]{\tiny $1$}
            (-3,0) node{\tiny{Toronto}}
            (-2.5,-3) node{\tiny{Midland}}
            (-2.5,-2) node[vtx]{\tiny{$12$}}

            (1.5,-3.5) node{\tiny{Guelph}}
            (.5,-2.5) node[vtx]{\tiny{$9$}}
            (-1,-4) node{\tiny{Rapids}}
            (-1,-2) node[vtx]{\tiny $11$}
            (2,0) node[vtx]{\tiny $4$}
            (2,-1) node{\tiny{Chicago}}
            (0,2) node[vtx]{\tiny $6$}

            (1.5,6) node{\tiny{LA}}
            (2,4) node[vtx]{\tiny $7$}

            (-3,2) node{\tiny{D.C}}
            (4,2) node[vtx]{\tiny $5$}
            (6,0) node{\tiny{Boulder}}
            (6.5,2) node{\tiny{Vancouver}}
            (4,0) node[vtx]{\tiny $8$}

            (4,6) node{\tiny{Seattle}}
            (4,4) node[vtx]{\tiny $2$}

            (-3,4) node{\tiny{Puebla}}
            (0,4) node[vtx]{\tiny $3$}

            (6,4) node[vtx]{\tiny{$10$}}
            (7.5,3) node{\tiny{Nanaimo}}
            (6.5,7) node{\tiny{Eureka}}
            (6,6) node[vtx]{\tiny{$13$}};

        \end{tikzpicture}
    \end{minipage}
    \qquad
    \begin{minipage}[b]{.2\textwidth}
        \begin{tikzpicture}[scale=.3,inner sep=1pt]
            \tikzstyle{vtx}=[circle,draw,thick,fill=white]
            \draw[thick] (-2,-2)--(0,0)--(2,0)--(4,0)--(4,2)--(4,4)--(2,4)--(0,4);
            \draw[thick] (-1,-2)--(0,0);
            \draw[thick] (0,-2)--(0,0);
            \draw[thick] (0,4)--(0,2)--(0,0);
            \draw[thick] (6,6)--(4,4)--(6,4);
            \draw (0,0) node[vtx]{\tiny $1$}
            (-3,0) node{\tiny{Toronto}}
            (-3,-3) node{\tiny{Midland}}
            (-2,-2) node[vtx]{\tiny{$2$}}

            (1,-3.5) node{\tiny{Guelph}}
            (0,-2.5) node[vtx]{\tiny{$4$}}
            (-1,-4) node{\tiny{Rapids}}
            (-1,-2) node[vtx]{\tiny $3$}
            (2,0) node[vtx]{\tiny $5$}
            (2,-1) node{\tiny{Chicago}}
            (0,2) node[vtx]{\tiny $13$}

            (1.5,6) node{\tiny{LA}}
            (2,4) node[vtx]{\tiny $11$}

            (-3,2) node{\tiny{D.C}}
            (4,2) node[vtx]{\tiny $7$}
            (6,0) node{\tiny{Boulder}}
            (6.5,2) node{\tiny{Vancouver}}
            (4,0) node[vtx]{\tiny $6$}

            (4,6) node{\tiny{Seattle}}
            (4,4) node[vtx]{\tiny $8$}

            (-3,4) node{\tiny{Puebla}}
            (0,4) node[vtx]{\tiny $12$}

            (6,4) node[vtx]{\tiny{$9$}}
            (7.5,3) node{\tiny{Nanaimo}}
            (6.5,7) node{\tiny{Eureka}}
            (6,6) node[vtx]{\tiny{$10$}};

        \end{tikzpicture}
    \end{minipage}
    \begin{minipage}[b]{.6\textwidth}

        \begin{tikzpicture}[scale=.6,inner sep=1pt,every node/.style={anchor=west,font=\small}]
            \matrix [nodes={draw,fill=blue!3,minimum size=2.5mm}] {
                \node{\tiny T}; & \node{\tiny{M}}; & \node{\tiny T}; & \node{\tiny R}; & \node{\tiny T}; & \node{\tiny G}; & \node{\tiny T}; & \node{\tiny C};                          & \node{\tiny C}; & \node{\tiny B}; & \node{\tiny V}; & \node{\tiny B}; & \node{\tiny S}; & \node{\tiny N}; & \node{\tiny S}; & \node{\tiny E};
                                & \node{\tiny S};                 & \node{\tiny L}; & \node{\tiny L}; & \node{\tiny P}; & \node{\tiny D}; & \node{\tiny P}; & \node{\tiny D};                          & \node{\tiny T};                                                                                                                                 \\
                \node{\tiny M}; & \node{\tiny T};                 & \node{\tiny R}; & \node{\tiny T}; & \node{\tiny G}; & \node{\tiny T}; & \node{\tiny C}; & \node{\tiny T};                          & \node{\tiny B}; & \node{\tiny C}; & \node{\tiny B}; & \node{\tiny V}; & \node{\tiny N}; & \node{\tiny S}; & \node{\tiny E}; & \node{\tiny S}; &
                \node{\tiny L}; & \node{\tiny S};                 & \node{\tiny P}; & \node{\tiny L}; & \node{\tiny P}; & \node{\tiny D}; & \node{\tiny T}; & \node{\tiny D}; \\};

            \draw (11,.75) node{\tiny{Edge list}};
            \draw (11,.25) node{\tiny{I}};
            \draw (11,-.25) node{\tiny{J}};
        \end{tikzpicture}

        \begin{tikzpicture}[scale=.6,inner sep=1pt,every node/.style={anchor=west,font=\small}]
            \matrix [nodes={draw,fill=yellow!3,minimum size=2.5mm}] {
                \node{\tiny T}; & \node{\tiny M}; & \node{\tiny R}; & \node{\tiny G}; & \node{\tiny C}; & \node{\tiny B}; &
                \node{\tiny V}; & \node{\tiny S}; & \node{\tiny N}; & \node{\tiny E}; & \node{\tiny L}; & \node{\tiny P}; & \node{\tiny D}; \\
            };
            \draw (11,.05) node{\tiny{Uniquified edge list}};\\
            \draw (11,-.45) node{\tiny{New ordering}};\\
            \draw (.15,-.45) node{\tiny{1}};
            \draw (.60,-.45) node{\tiny{2}};
            \draw (1.0,-.45) node{\tiny{3}};
            \draw (1.4,-.45) node{\tiny{4}};
            \draw (1.85,-.45) node{\tiny{5}};
            \draw (2.3,-.45) node{\tiny{6}};
            \draw (2.75,-.45) node{\tiny{7}};
            \draw (3.2,-.45) node{\tiny{8}};
            \draw (3.65,-.45) node{\tiny{9}};
            \draw (4.00,-.45) node{\tiny{10}};
            \draw (4.45,-.45) node{\tiny{11}};
            \draw (4.9,-.45) node{\tiny{12}};
            \draw (5.35,-.45) node{\tiny{13}};
        \end{tikzpicture}

    \end{minipage}
    \caption{
        An illustration of the BOBA algorithm vs.\ degree ordering over a small, almost uniform graph, that we pretend represents some roads in North America. The input COO graph is shown in blue with $I$ above $J$. Toronto and Seattle have degree 5 and 4 respectively, and the rest of the vertices have degree 2 or 1.
        \textbf{Left}: Order-by-degree can be arbitrarily bad in terms of spatial locality since vertices with the same degree will appear in arbitrary order. %Toronto and Seattle come first and second, but they are half a continent apart; why should they be placed near each other?
        \textbf{Right}: The BOBA order constructed from the uniquified edge list is shown at the bottom in yellow, and the new numeric labels below. Note pairs of vertices connected by an edge are reasonably close together in the resulting order. Better spatial locality in the cache, then, should improve traversal performance.
    }\label{fig:road}
\end{figure}
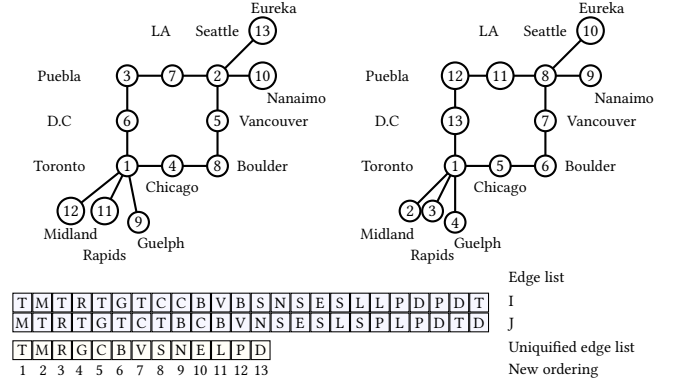

\section{Preliminaries}
\label{sec:prelim}
\subsection{Graphs}
In this paper, we are concerned with directed graphs $G$ with edges and vertices $E(G)$ and $V(G)$\@. The relevant \defn{neighborhood} will depend on whether the graph application is \emph{push} or \emph{pull} implementation. For our analysis, we will assume the push direction, i.e., $N(v) := \{u \in V(G) \mid vu \in E(G) \}$, and the \defn{degree} $\deg(v) = |N(v)|$. A \defn{coordinate representation (COO)}, which we will sometimes write as $\text{COO}(G) := (I,J)$, is a listing of the \emph{directed} edges of $G$ given by pairs of vectors $(I,J)$ where $I,J \in [n]^{m}$, with $m:= |E(G)|, n:=|V(G)|$, such that $E(G) = \{I(i)J(i) \mid i \in [m]\}$.
We will often identify $v_i \in \{v_1,v_2,\ldots,v_n\}$ with its index $i \in [n]$. All graphs considered in this paper are assumed to be \defn{sparse}, that is, there is some constant $C$ such the number of edges is $Cn \ll n^2$.

\section{Overview of Sparse Graph Reordering}
\label{sec:prev}
Generally sparse graphs are represented in memory as sparse matrices. The terms `graph' vs.\ `matrix' reordering typically refer to optimizing for particular  access patterns, rather than any ability to apply the reordering to the object. Matrix reorderings may apply to, say, numerical methods such as matrix decompositions and iterative solvers, and graph reorderings for more irregular graph-traversal-based algorithms. In practice the goals of both have significant overlap; for instance, PageRank is considered a graph algorithm, even though it is typically implemented with repeated applications of \spmv, a canonical matrix-based algorithm.

In 1999, \citet{anderson:1999:achieving} used matrix reordering to compute \spmv in a large cluster where the matrix is typically partitioned among nodes, and each part is locally reordered using a technique such as Reverse Cuthill-McKee (RCM)~\cite{Cuthill:1969:RTB} to reduce bandwidth and increase cache coherency.  While reordering algorithms are often first presented with sequential formulations, in 2014 \citet{Karantasis:2014:PRA} gave a parallelization of both RCM~\cite{Cuthill:1969:RTB} and Sloan~\cite{sloan:1986:algorithm}. \spmv was selected because, as mentioned, it is a key computation in important iterative methods, such as sparse solvers, and page rank. These parallel implementations of RCM and Sloan, though somewhat complex, produce reorderings of similar quality to the original sequential versions, and reduce the number of iterations of \spmv required to realize an end-to-end speedup. Still, because of the relatively high cost of the reordering, these techniques typically require on the order of hundreds of iterations to achieve speedup. Our goal for BOBA is to produce an ordering, using a parallel-friendly method, as part of the graph creation process that gives an \emph{end-to-end speedup in graph construction and \spmv after one iteration}.

\subsection{Objective-Based Approaches}
The first algorithms for matrix reordering modeled the problem of mapping vertices to a cache-friendly ordering as an optimization problem. This approach is NP-hard, so either heuristics or approximation algorithms are used to find feasible solutions. \defn{Heavyweight reordering} techniques target Problem~\ref{prob:offline}, with the quality of the result more important than the time to reorder. Such techniques assume the application speedup achieved from increased cache utilization will pay off in aggregate after many runs, the number of which is not a consideration in the design of the reordering algorithm.

\subsubsection{Bandwidth}\label{model:bandwidth}
The \defn{bandwidth} of a matrix is the maximum distance of a non-zero element to its diagonal. The \textsc{Bandwidth Problem} for a graph $G$ was shown to be NP-hard by Christos Papadimitriou~\cite{papadimitriou1976np}:
\[\argmin_{p} \max_{uv \in E(G)}|p(u) - p(v)|\]

\noindent
RCM~\cite{Cuthill:1969:RTB} is a commonly used heuristic algorithm for this objective with runtime $O(\deg_{\text{max}}|E|)$~\cite{liu:1976:reducing}.
Graph-contraction-based approaches, such as node dissection and Approximate Min Degree (AMD)~\cite{amestoy1996approximate}, have been employed for decades by the numerical community in factoring sparse matrices. These seem to often give good results for graph algorithms as well; however, these approaches require edge deletion and vertex contraction operations, which are a poor fit for typical static sparse graph formats like CSR\@.

\subsubsection{Gorder}
In 2016, Wei et al.\ present \defn{Gorder}~\cite{Wei:2016:SGP} as a method to approximate a vertex ordering $p_{\text{GO}}(V) := v_{1}v_{2}\ldots v_n$ that maximizes an NP-hard TSP formulation. The number of shared neighbors and edges between vertices that are placed within distance $w$ of each other are maximized. If $w$ represents the cache line size, they argue that this is a good proxy for cache alignment.
\begin{model}[GScore]\label{model:gscore}
  Given ordered graph $G$, define $s(u,v) := |(N(u) \cap N(v)| + |\{uv,vu\} \cap E(G)|$, then
  \textsc{$\text{GScore}(G,w) := \sum_{i=1}^{n}\sum_{j = \max(1,i-w)}s(v_i,v_j)$.}
\end{model}

They present a $\frac{1}{2w}$-factor approximation algorithm, which they show tends to do even better on real-world data sets. This algorithm runs in time $O(w \deg_{\text{max}} n^2)$, where $\deg_{\text{max}}$ is the maximum degree of $G$~\cite{Wei:2016:SGP}. In practice, on graphs with billions of edges, Gorder takes hours to run and can require thousands of iterations of the graph application in order to break even~\cite{Chen:2022:WBG}.

\subsection{Lightweight Approaches}
\label{sec:previous}
Lightweight approaches, which target Problem~\ref{prob:online}, are orders of magnitude faster than heavyweight approaches, but generally result in a more modest post-reorder speedup. Existing lightweight methods focus on \defn{power law} graphs. Empirical studies (e.g., \citet{eubank2004structural,newman2018networks}) have shown that the degree distribution in many real-world networks crudely follows a power law. Intuitively, this means the network is sparsely structured into a small number of \defn{hub} vertices that are adjacent to most of the other (small-degree) vertices. To the best of our knowledge, all existing lightweight methods attempt to leverage this skewness property. We say a graph is \defn{skew} or \defn{scale-free} if its degree distribution follows the common pattern with a relatively few very high-degree \defn{hub} vertices, and the rest have low degree. More formally, Newman~\cite{newman2018networks} suggests the pure power law as a simple way to model such distributions: the number of vertices of degree $k$ is roughly $\frac{n}{\zeta(\alpha)k^{\alpha}}$ for some $\alpha \in [2,4]$.

One natural idea in skew distributions is to \defn{sort by reverse degree}. This places all hub vertices at the beginning, with hopes that they form a densely connected subgraph. Clearly the first cache block of such an ordering would contribute well to both the NScore (\secref{model:nscore}) and the bandwidth (\secref{model:bandwidth}) objectives.  Note that if the degree distribution is not skew, but instead more uniformly distributed, sort by degree is essentially the same as taking a random permutation of vertices. See Figure~\ref{fig:road}.

Even with a skew distribution, if the high-degree subgraph is not densely connected or the original ordering already has most hubs located close together, then degree-based approaches can actually reduce cache effectiveness. Balaji and Lucia~\cite{balaji2018graph} give a metric to help predict if the input graph will benefit from degree-based lightweight reordering approaches. Rabbit Ordering by Arai et al.~\cite{arai2016rabbit} is another lightweight technique, that is based on community detection.
These degree based methods deliver end-to-end speedups in applications such as PageRank in scale-free networks. To reduce the runtime of a full sort, partial sorting methods have been devised such as frequency-based sorting (also known as hub-sort)~\cite{zhang2017making}, hub-cluster~\cite{balaji2018graph}, and degree binning~\cite{faldu2019closer}. The main idea is to separate or sort only \emph{hub} vertices. % which are defined to have at least average degree ($v \in V(G)$ such that $\deg(v) \geq \frac{m}{n}$).

We have focused on workflows where the input graph has random, or even non-numeric, labels. We measured our success by comparing our graph-algorithm runtimes against the same graph algorithms run on randomly labeled inputs. Since the process of randomization eliminates any inherent structure of the original labeling, we did not compare against existing lightweight methods that rely on preserving such structure, such as degree-based grouping~\cite{faldu2019closer} or hub clustering~\cite{balaji2018graph}, but we have included hub sort~\cite{zhang2017making} and full sort by degree, which do not make these assumptions.

\subsection{Reordering in Parallel Graph Processing}
Chen and Chung~\cite{Chen:2022:WBG} have recently presented a cache-aware reordering targeting the multi-core parallel setting. They argue that load balancing becomes the dominant factor in determining graph processing time.
Indeed, \defn{vertex-centric} approaches, in which threads, or blocks of threads, are initially assigned to vertices is not an uncommon approach. Unfortunately, in skew networks, this natural approach suffers from \defn{load imbalance}, as the threads assigned to hub vertices will be tasked with a disproportionate amount of work. See, e.g., the book by Sanders et al.~\cite{sanders2019sequential}. In fact, degree-based reorderings can exacerbate this issue, as the first block of threads will have to process all hub vertices.

To overcome load imbalance, \defn{edge-centric} approaches evenly divide edges, and thus work, across threads. For instance, Merrill and Garland~\cite{Merrill:2016:MPS} give a perfectly balanced \emph{merge-path}~\cite{Green:2012:GMP} approach for \spmv. In our evaluation, we also use merge-path load balancing in our algorithmic benchmarking implementations. Since BOBA is much less precise globally than degree-based sorting methods, and reorders instead based on local edge structure, it doesn't tend to result in the types of cache optimizations that make load-balancing dramatically worse.

\section{The Algorithm and Analysis}
\label{sec:algorithm}
Whereas RCM and Sloan are heuristic solutions to the bandwidth problem (\secref{model:bandwidth}) and Gorder~\cite{Wei:2016:SGP} approximates a generalization of TSP (\secref{model:gscore}), BOBA is not explicitly based on an NP-hard optimization problem. Instead of directly approximating an objective, BOBA takes its inspiration from the \defn{preferential attachment}(PA) process defined by Albert and Barab{\'a}si~\cite{albert2002statistical}, which has been used extensively to model real-world sparse networks. Still, in Proposition~\ref{prop:regular}, we show that under pristine conditions, when the input is a $d$-regular COO graph, sorted by destination, that BOBA actually is a $d+1$ factor approximation algorithm, for a very similar TSP formulation as Gorder (\secref{model:gscore}). Therefore, BOBA has a theoretical performance guarantee in these types of COO graphs: it is not be worse than $\frac{1}{d+1}$ of the optimal solution to our proxy for cache coherency. We call our (slightly) simplified TSP optimization model \defn{neighbor score} (NScore) Model~\ref{model:nscore}.

\subsection{A Proxy for Cache Coherency}
We are interested in scoring permutations with some theoretical measure. For this purpose we use an objective that is similar to the GScore$(w)$ TSP objective Model~\ref{model:gscore} developed by \citet{Wei:2016:SGP}. For the clarity of our analysis, we restrict ourselves to the case where $w=1$ (a cache size of 2) but the intuition should generalize to larger cache sizes. Recall that the neighborhood of $v$ is $N(v) := \{u \in V(G) \mid vu \in E(G)\}$. The objective is to find an ordering of vertices $p := p_1,p_2,...p_n$, where
$\{p_1,p_2,\ldots,p_n \} = V(G)$, that maximizes

\begin{model}[NScore]\label{model:nscore}
    Given graph $G$, and ordering $p$ define
    \[\text{NScore$(G,p)$} := \sum^{n-1}_{i=1} |N(p_i) \cap N(p_{i+1})|.\]
\end{model}

For a given $G$, let $p^*$ be an optimal ordering with respect to NScore. That is $p^*$ maximizes NScore$(G,p^*)$. A quick observation yields a coarse upper bound.
\begin{lemma}\label{lem:upperbound}
    Let $G$ be a graph with $n$ vertices and $m$ edges. Then NScore$(G,p^*) \leq m$.
\end{lemma}
\begin{proof}
    From the definition:
    \begin{align*}
        \text{NScore$(G,p^*)$} & = \sum^{n-1}_{i=1} |N(p^*_i) \cap N(p^*_{i+1})|,        \\
                               & \leq \sum^{n-1}_{i=1} \min(|N(p^*_i)|, |N(p^*_{i+1})|), \\
                               & \leq \sum_{v\in V(G)} \deg(v),                          \\
                               & = m.
    \end{align*}
\end{proof}

Before we formally state the algorithm, we analyze the natural choice of ordering vertices by their attachment time in the preferential attachment network model.
\subsection{Preferential Attachment}
Let $G^n_c$ be an undirected graph with $nc$ edges created by the LCD model of preferential attachment described by Bollob\'{a}s and Riordan~\cite{Bollobas:2004:TDO}. We do not restate the LCD model here. We simply point out that the graph is created by running $c$  $G^n_1$ processes, in which vertex $v_t$ added at time $t$ randomly selects an existing vertex $v_s$ in proportion to $v_s$'s degree at time $t-1$, and attaches to that vertex. That is, an edge $v_t \rightarrow v_s$ is attached with probability $\frac{\deg_{t-1}(v_s)}{2t-1}$, where $2t-1$ is the number of edges in $G_1^{t-1}$. We then form a single graph $G^n_c$ on $n$ vertices with $nc$ edges by running the $G^n_1$ process $c$ times, then identify vertices by arrival time.

We state a trivial corollary that immediately follows from the proof of Theorem 16 in Bollob\'{a}s and Riordan~\cite{Bollobas:2002:MRO}:

\begin{corollary}[Corollary of Theorem 16~\cite{Bollobas:2002:MRO}] \label{lem:bollobas}
    Let $G_c^n$ be a graph generated via the LCD model of preferential attachment, and let $p := p_1,p_2,\ldots,p_n$ be any ordering of $V(G)$.
    Then \[\mathbb{E}(\text{NScore}(G^n_c,p)) \approx \sum_{i = 1}^{n-1}\sum_{1 \leq a \leq \min(p_i,  p_{i+1})} \frac{c(c+1)}{a\sqrt{p_{i}p_{i+1}}}.\]
\end{corollary}

From this result, we see that $\mathbb{E}(\text{NScore}(G^n_c,p))$ is maximized by setting $p$ to the identity, $p = 1,2,3,\ldots,n$. In other words, we see for $G_c^n$ that we can not expect to do much better than ordering by attachment time.

\subsection{The Sequential Algorithm}

\begin{algorithm}\caption{\sc{Sequential BOBA}}\label{alg:seq}
    \begin{scriptsize}
        \begin{algorithmic}[H]
            \REQUIRE a graph $G$ such that each vertex is in at least one edge.
            \ENSURE a permutation $p$ of $V(G)$.
            \STATE{$(I,J) \gets $ COO$(G)$, $S \gets \varnothing$, $i \gets 0$, $p \gets \mathbf{0}^{n}$, $m \gets |E(G)|$}
            \FOR{$ j : 1 \ldots m$}
            \STATE{$v \gets I(j)$}
            \IF{$v \notin S$}
            \STATE{$p_i \gets v$; $i \gets i + 1$; $S \gets S \cup \{v\}$}
            \ENDIF
            \ENDFOR
            \IF{$i = n$}
            \STATE{return $p$}
            \ELSE
            \FOR{$j: 1 \ldots m $}
            \STATE{$v \gets J(j)$}
            \IF{$u \notin S$}
            \STATE{$p_i \gets u$; $i \gets i + 1$; $S \gets S \cup \{u\}$}
            \ENDIF
            \ENDFOR
            \ENDIF
            \STATE{return $p$}
        \end{algorithmic}
    \end{scriptsize}
\end{algorithm}
Algorithm~\ref{alg:seq} builds a permutation $p$ such that if $v = p_k$, then $k$ is the lowest index that $v$ appears in $I$ logically appended to $J$ (denoted $I++J$). It scans through $I++J$, maintaining a set, and counter of unique vertices seen so far, and updates $p$ until $n$ unique vertices have been encountered.

Figure~\ref{fig:road} illustrates why degree based methods tend to be inappropriate for road and similar types of naturally occurring networks. On the other hand, we can formally show that Algorithm~\ref{alg:seq} gives at least a bounded performance guarantee.
\subsection{Approximation Guarantee for $d$-regular Graphs}
Road-type networks in the real world tend to have low, fairly uniform degree distributions. Here we show that if a graph is $d$-regular, meaning each vertex has degree $d$, and the input COO graph is sorted, then Algorithm~\ref{alg:seq} gives an ordering within a factor of $\frac{1}{d+1}$ from optimal. Note that Gorder gives a $\frac{1}{2}$ guarantee, and so for low values of $d \in \{2,3,4,5\}$, BOBA's guarantee is a relatively small constant factor worse than Gorder's.

\begin{prop}\label{prop:regular}
    Let $G$ be a $d$-regular graph, $d > 1$, and $p_B$ be the ordering returned by BOBA(COO$_A(G)$), where COO$_A(G)$ lists the edges of $G$ sorted by destination. \\
    Then $(d+1)$NScore$(G,p_B) \geq$ NScore$(G,p^*)$.
\end{prop}
\begin{proof}
    Let $s(d,m)$ be defined as the NScore of the permutation returned by BOBA on a COO graph sorted by destination where each vertex appears as a source $d$ times. The first $d$ edges will be reordered as $(1,x),(2,x),\ldots,(d,x)$ where $x$ is the first destination in the edge list. Notice that these edges contribute at least $d-1$ to NScore$(G,p_B)$. Next, remove all edges with destinations adjacent to sources in $\{1,2,\ldots,d\}$ from the edge list.
    Since $G$ is $d$-regular, this removes at most $d^2$ edges, and each remaining source vertex still appears in $d$ edges. Thus we have
    \begin{align*}
        s(d,m) & \geq d-1 + s(d,m-d^2),                      \\
               & \geq \frac{(d-1)m}{d^2},                    \\
               & \geq \frac{d-1}{d^2}\text{NScore$(G,p^*)$},
    \end{align*}
    where the last line is given by Lemma~\ref{lem:upperbound}.
\end{proof}

These few theoretical insights conclude our treatment of the sequential Algorithm~\ref{alg:seq}.
We now turn our attention to designing for large real-world data-sets for which the precise conditions of Proposition~\ref{prop:regular} and Corollary~\ref{lem:bollobas} are unrealistic, and which require as much speed as possible. We now present parallel Algorithm~\ref{alg:par}.

\begin{algorithm}\caption{\sc{BOBA}}\label{alg:par}
    \begin{scriptsize}
        \begin{algorithmic}[H]
            \REQUIRE a graph $G$ such that each vertex is in at least one edge.
            \ENSURE a permutation $p$ of $V(G)$.
            \STATE{$(I,J) = $ COO$(G)$;  $r \gets \mathbf{\infty}^{n}, p \gets (1,2,\ldots,n)$}
            \FORALL{$ i \in [2m]$ in parallel }
            \IF{$i < m$ and $i < r(I(i))$} \label{line:cond1}
            \STATE{$r(I(i)) \gets i$} \label{line:assign1}
            \ELSIF{$i < r(J(i))$} \label{line:cond2}
            \STATE{$r(J(i)) \gets i$} \label{line:assign2}
            \ENDIF
            \ENDFOR
            \STATE{// $r(i)$ uniquely corresponds to the position of $i$ in $p$.}
            \STATE{return \textsc{ParMapKeys}$(p,r)$} \label{line:mapbykey}
        \end{algorithmic}
    \end{scriptsize}
\end{algorithm}

\noindent
Algorithm\ref{alg:par} is similar in spirit to Algorithm\ref{alg:seq}, it finds a set of $n$ indexes $r$, in $[2m]$, corresponding to $n$ unique vertices in $I++J$, however, the index picked corresponding to vertex $r(k) = v$ is no longer guaranteed to be the lowest index where $v$ appears in $I++J$. We could recover this sequential property by using an \texttt{AtomicMin} operation at lines~\ref{line:assign1} and~\ref{line:assign2} (at the cost of reordering performance), we found in practice that the resulting permutation did not yield reorderings that delivered significantly better performance. Since $r \in [2m]^n$, and each dimension is a unique key associated with each element of $p$, line~\ref{line:mapbykey} can be accomplished in $O(n)$ time (e.g., hash tables~\cite{Awad:2021:BGH}).

\section{Results}
\label{sec:results}

\subsection{Experimental Setup}

\paragraph{Hardware} We evaluate our implementations on an NVIDIA DGX Station with V100 GPUs with 32~GB DRAM and an Intel Xeon CPU E5-2698 v4.  The GPU has a 6~MiB L2 cache, 80 streaming multiprocessors (SM), and a per-SM 128~KiB L1 cache. We perform all CPU-based reorderings on an AMD EPYC 7402 24-Core Processor CPU\@. The GPU has a theoretical achievable DRAM bandwidth of 897~GiB/s. Our code is complied with CUDA~11.6.

\paragraph{Datasets and reordering} Table~\ref{tab:datasets} summarizes the datasets we use in our benchmarks. We only used the undirected datasets for implementations that require an undirected graph (e.g., triangle counting). We randomized the datasets listed, generated, and report times for, hub sort and degree order using the reordering tool given in the companion repository\footnote{\url{https://github.com/CMUAbstract/Graph-Reordering-IISWC18}} for Balaji et al.~\cite{balaji2018graph}. For Gorder~\cite{Wei:2016:SGP}, we used the authors' repository.\footnote{\url{https://github.com/datourat/Gorder}} For RCM we used MATLAB~\cite{MATLAB:R2021a}.

\paragraph{Applications}
We evaluate BOBA and other reordering techniques on four different graph algorithms (sparse matrix-vector multiply, PageRank, triangle counting, and single-source shortest path), each featuring a different type of graph traversal. We evaluate performance on the GPU graph framework of \citet{Osama:2022:EOP} using \emph{merge-path} load balancing~\cite{Green:2012:GMP} for all graph algorithms and datasets.

SpMV performs a neighbor-reduce operation to compute the dot product between each row and the input vector, then stores the result in the output vector. Given the sparsity of each row, a good reordering scheme will improve the access to the input vector, achieving coalesced accesses while reading the vector.

In PR, each edge's source propagates its weight to its neighbor vertices using an atomic operation. As with both TC and SSSP, the atomic operation will benefit from an efficient reordering algorithm. PR differs from the other algorithms since it operates on the entire graph multiple times until convergence.

In TC, for each edge in the graph, we perform a set-intersection operation between the adjacency lists of the edge source and destination vertices. The edge source adjacency list will be already in the cache. However, the destination vertex may or may not be in the cache. An efficient reordering algorithm will reorder the vertices such that adjacent threads (or a group of threads) process neighbor vertices, which facilitates a larger hit rate in the L1 or L2 cache when loading the adjacent list of the destination vertex. During the set-intersection operation, we atomically increment the triangle count of the common vertices. If the neighbor vertices have adjacent labels, the atomic operations will be stored in the same cache line, facilitating efficient (almost) coalesced atomic operations.

SSSP features sparse frontiers of vertices, atomic updates to destination vertices' distances, and traversal of neighbor vertices. Good reorderings improve cache locality for neighboring adjacency lists and neighboring distance values.

\paragraph{Summary of results}  BOBA improves cache locality on both CPUs and GPUs, achieving application speedups similar to heavyweight techniques. On \spmv it achieves speedups over random ranging from 1.17--6.25$\times$ with median 3.5$\times$ in skew networks, and 2.25--5.5$\times$ with median 3.4$\times$ in road-like networks. For end-to-end \spmv, BOBA achieves speedups from 1.36--3.22$\times$ with median of 2.35$\times$.

Compared to randomly labeled graphs, BOBA-reordered graphs achieve end-to-end speedups (including reordering time) up to 3.45$\times$. The reordering time is approximately one order of magnitude faster than existing lightweight techniques, and up to 2.5 orders of magnitude faster than heavyweight techniques. Because of contention, BOBA doesn't improve TC, achieving $\sim$0.6 slower than random; however, BOBA achieves L1 hit rates between 40--95\% for TC\@. For all applications, BOBA also achieves hit rates on par with heavyweight methods as measured during runtime (e.g., for \spmv, 7--52\% and 11--67\% in the L1 and L2 caches, respectively), and by our own spatial locality metrics (0.48 to 0.88).

\subsection{Our NBR Spatial Locality Metric}

As a metric for measuring spatial locality, we define

$\textit{NBR}(G)$ as the expected ratio of cache lines to neighbors of a randomly selected vertex.

\[ \textit{NBR}(G) := \frac{1}{n}\sum_{v \in V(G)}\frac{\text{lines spanned by }N(v)}{|N(v)|}.\]

\noindent
Lower $\textit{NBR}(G)$ implies better spatial locality. Table~\ref{tab:NBR32} shows that Gorder (a heavyweight method) consistently has the best NBR but BOBA is slightly better than RCM and much better than Hub.

\begin{table}
    \setlength{\belowcaptionskip}{0in}
    \setlength{\abovecaptionskip}{0in}

    \captionsetup{font=footnotesize,labelfont=footnotesize}
    \resizebox{\columnwidth}{!}{\begin{tabular}{lccccccccccc}
            \toprule
            \multirow{2}{*}{Datasets} & \multicolumn{1}{c}{Rand} & \multicolumn{1}{c}{Gorder} & \multicolumn{1}{c}{RCM } & \multicolumn{1}{c}{BOBA} & \multicolumn{1}{c}{Hub} \\                               & \multicolumn{1}{l}{NBR} & & & & \\ \midrule
            delaunay\_n23             & 0.99                     & \textbf{0.38}              & 0.56                     & 0.48                     & 0.99                    \\
            delaunay\_n24             & 0.99                     & \textbf{0.38}              & 0.56                     & 0.48                     & 0.99                    \\
            hollywood-2009            & 0.99                     & \textbf{0.49}              & 0.59                     & 0.57                     & 0.94                    \\
            great-britain\_osm        & 1.0                      & \textbf{0.60}              & 0.93                     & 0.65                     & 0.99                    \\
            road\_usa                 & 1.0                      & \textbf{0.65}              & 0.85                     & 0.79                     & 0.99                    \\
            arabic-2005               & 0.84                     & \textbf{0.25}              & 0.40                     & 0.43                     & 0.84                    \\
            soc-LiveJournal           & 0.88                     & \textbf{0.64}              & 0.79                     & 0.77                     & 0.88                    \\
            ljournal-2008             & 0.89                     & \textbf{0.65}              & 0.79                     & 0.76                     & 0.88                    \\
            kron\_g500-logn20         & 0.75                     & \textbf{0.70}              & 0.72                     & 0.74                     & 0.72                    \\
            kron\_g500-logn21         & 0.72                     & \textbf{0.69}              & 0.71                     & 0.72                     & 0.70                    \\
            soc-orkut                 & 0.99                     & \textbf{0.75}              & 0.89                     & 0.88                     & 0.96                    \\
            \bottomrule
        \end{tabular}}
    \caption{NBR Line metric over CSR for our graph datasets for the different reordering schemes. Lower metrics indicate better spatial locality.}
    \label{tab:NBR32}
\end{table}

\subsection{End-to-end Time}

We now evaluate Problem~\ref{problem:pragmatic-graph-reordering}: can we achieve an end-to-end speedup on a graph algorithm including reordering time? End-to-end time includes the time to reorder the COO, convert the COO input graph to CSR, and to run the graph algorithm. We assume that input labels are already randomized. Additionally, we add the cost of sorting the COO input graph for the TC algorithm; in our implementation, converting a sorted COO graph results in a CSR graph with a sorted per-vertex adjacency list (which is required for the set-intersection operation in TC). We implement both conversion and sorting on the CPU\@.

Figure~\ref{fig:e2e} compares BOBA's performance to randomly labeled datasets. Except for TC, the cost of converting COO to CSR dominates overall runtime. BOBA significantly improves this conversion time, achieving speedups between \{1.3, 4.8\} for SpMV, \{1.3, 5.1\} for PR, and \{1.3, 3.3\} for SSSP, respectively. Improvements in the CPU-based COO to CSR conversion are the result of BOBA improving cache locality on the CPU\@.

Interestingly, sorting the input COO graph in TC almost eliminates the conversion cost for both BOBA and random; however, sorting is very expensive. For instance, sorting delaunay\_24 is 10.5$\times$ and 13$\times$ slower than converting BOBA-reordered and random-labeled graphs. In general, BOBA slightly improves the cost of sorting the COO graphs. BOBA achieves sorting speedups between 1.045 (hollywood-2009) to 1.54 (great-britain\_osm). For the end-to-end TC cost, BOBA is as high as 1.4$\times$ faster than random (great-britain\_osm); however, BOBA end-to-end time is 0.62$\times$ and 0.6$\times$ slower than random for the two kron graphs. We believe that the slowdown is due to increased contention. However, as we shall see in our cache analysis, BOBA nonetheless significantly improves the cache hit rate for TC\@.

\subsection{Reordering and Graph Algorithms Runtime}

Now we compare the time required for the reordering and graph algorithms. We compare BOBA reordering with two heavyweight reordering techniques (RCM, Gorder) and two lightweight ones (degree, hub-sort). We normalize all graph algorithm runtimes to random. Figures~\ref{fig:scale_free_pareto} and~\ref{fig:uniform_pareto} summarize our results for this benchmark.

\paragraph{Reordering time} Compared to heavyweight CPU-based graph reordering algorithms, our GPU-friendly BOBA algorithm achieves reordering times that are orders of magnitude faster. For scale-free datasets, except for the arabic-2005 dataset, reordering all datasets requires less than 100~ms, while other lightweight reordering algorithms are 100--2000~ms. Heavyweight algorithms always take more than 2000~ms and can be three orders of magnitude slower than BOBA (e.g., for kron\_g500-log21). For the arabic-2005 dataset, BOBA reordering (224 milliseconds) is also an order of magnitude faster than other lightweight algorithms ($\sim$3500 milliseconds) and more than 2.5 orders of magnitude faster than the heavyweight algorithms (77,000 and 93,500 milliseconds for Gorder and RCM, respectively). BOBA also achieves significant speedups for road-like graphs; however, compared to scale-free graphs, the cost of heavyweight reordering algorithms drop to less than 25,000 milliseconds.

\paragraph{Graph algorithm runtimes} After reordering, for the arabic-2005, hollywood-2009, and kron\_g500-log20 graphs, BOBA's SpMV performance is faster than all other heavyweight and lightweight techniques.  For instance, on arabic-2005 BOBA is slightly faster than GOrder, and RCM, but all methods give dramatic improvements over random. This correlates to our NBR metric Table~\ref{tab:NBR32}.

For the low clustering coefficient dataset kron\_g500-log21 dataset, NBR were not improved much from random, and improvements across all techniques are more muted. For all other scale-free graphs, BOBA-reordered graphs achieve a performance faster than degree-based techniques and slower than heavyweight ones.

For road-like graphs, as expected, degree-based reordering achieves performance that is either close to random (in SpMV, TC, and PR) or worse for a more expensive algorithm (SSSP)\@. In general, BOBA-ordered graphs achieve similar performance to the heavyweight techniques. Here we point to Proposition~\ref{prop:regular}, which is again consistent with Table~\ref{tab:NBR32}, with BOBA in a distinct second place on delauny, road\_usa, great-britain\_osm. All reordering techniques struggle with SSSP\@.

\begin{table}
    \setlength{\belowcaptionskip}{0in}
    \setlength{\abovecaptionskip}{0in}

    \captionsetup{font=footnotesize,labelfont=footnotesize}
    \resizebox{.5\columnwidth}{!}{\begin{tabular}{lccccc}
            \toprule
            \multirow{2}{*}{Datasets} & \multirow{2}{*}{$|V|$} & \multirow{2}{*}{$|E|$} & \multicolumn{2}{c}{Size in MB}                               \\
                                      &                        &                        & \multicolumn{1}{l}{offsets}    & \multicolumn{1}{l}{indices} \\ \midrule
            delaunay\_n22             & 4.2M                   & 25M                    & 16                             & 96                          \\
            delaunay\_n23             & 8.4M                   & 50M                    & 32                             & 192                         \\
            delaunay\_n24             & 16.8M                  & 100.7M                 & 64                             & 384                         \\
            great-britain\_osm        & 7.7M                   & 16.3M                  & 29.5                           & 62.2                        \\
            hollywood-2009            & 1.1M                   & 113.9M                 & 4.3                            & 434.5                       \\
            rgg\_n\_2\_22\_s0         & 4.2M                   & 60.7M                  & 16                             & 231.6                       \\
            rgg\_n\_2\_23\_s0         & 8.4M                   & 127M                   & 32                             & 484.5                       \\
            rgg\_n\_2\_24\_s0         & 16.8M                  & 265M                   & 64                             & 1,011                       \\
            road\_usa                 & 23.9M                  & 57.7M                  & 91.4                           & 220                         \\
            arabic-2005               & 22.7M                  & 639.9M                 & 86.6                           & 2,441                       \\
            kron\_g500-logn20         & 1M                     & 89M                    & 4                              & 340.4                       \\
            kron\_g500-logn21         & 2.1M                   & 182M                   & 8                              & 694.6                       \\
            soc-orkut                 & 3M                     & 212.7M                 & 11.4                           & 811.4                       \\
            soc-LiveJournal1          & 4.8M                   & 69M                    & 18.492                         & 263.2                       \\
            ljournal-2008             & 5.3M                   & 79M                    & 20.46                          & 301.4                       \\
            \bottomrule
        \end{tabular}}
    \caption{Datasets from this work are from the SNAP Network Analysis Project~\cite{leskovec:2016:snap} and the SuiteSparse Matrix Collection~\cite{Davis:2011:Florida}.}
    \label{tab:datasets}
\end{table}

\begin{figure}
    \centering
    \captionsetup{font=footnotesize,labelfont=footnotesize}
    \includegraphics[width=\columnwidth]{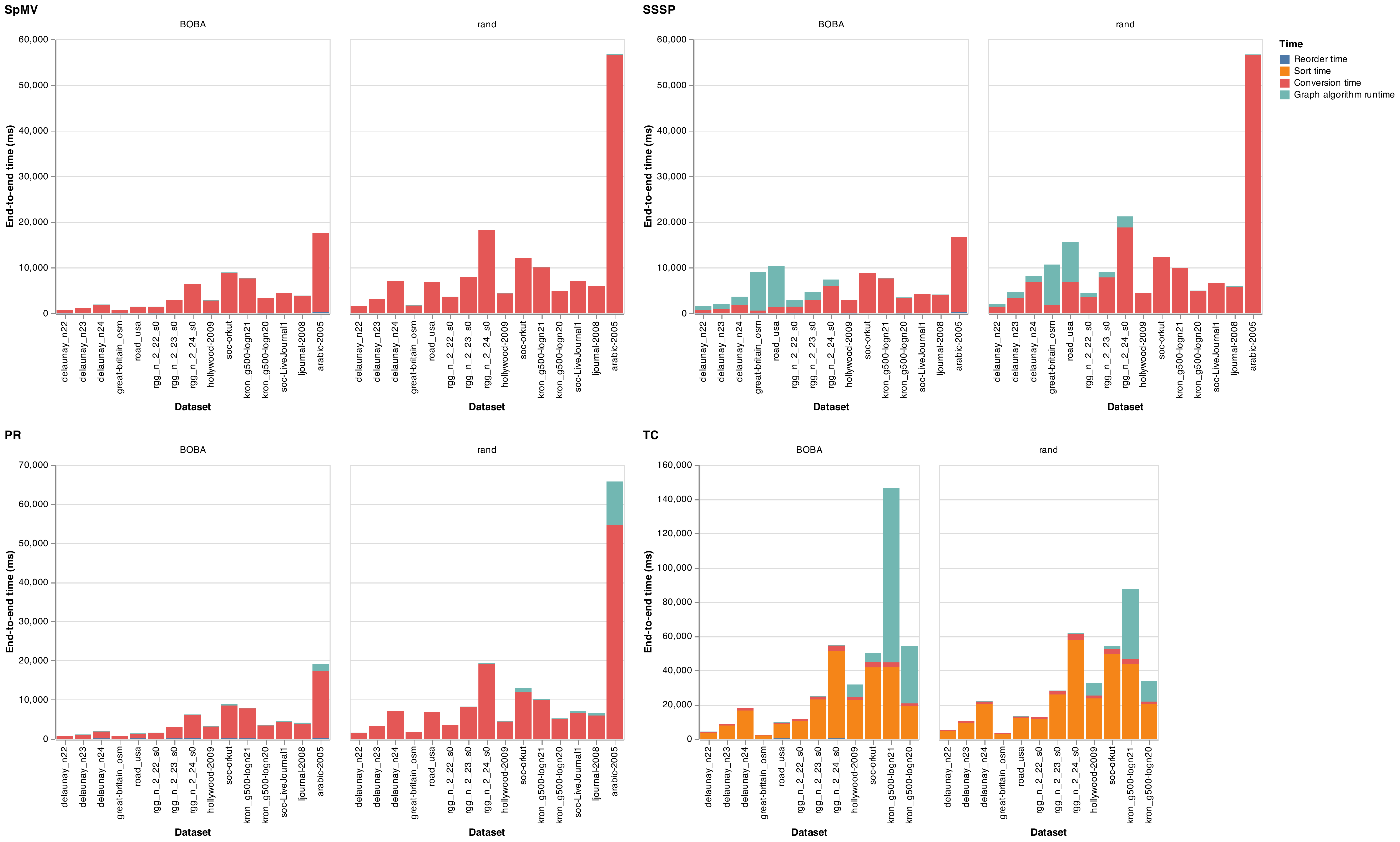}
    \caption{End-to-end results. Note that we sort the graph adjacency lists in only the TC algorithm.}
    \label{fig:e2e}
\end{figure}

%DATASETS_NAMES=("rgg_n_2_24_s0" "rgg_n_2_23_s0" "rgg_n_2_22_s0"
%                "hollywood-2009"
%               "delaunay_n24"  "delaunay_n23"  "delaunay_n22"
%              "great-britain_osm" "road_usa" "arabic-2005"
%             "soc-LiveJournal1" "ljournal-2008"
%           "kron_g500-logn20" "kron_g500-logn21"
%            "soc-orkut")

\subsection{Cache Hit Rate Analysis}

We hypothesized that the performance improvement from reordering comes from better cache performance. Thus we profile our graph applications with BOBA and other reordering techniques to evaluate cache performance. We measure cache hit rates at the different cache hierarchy levels (L1 and L2) and the percentage of memory transactions fulfilled by the GPU's DRAM\@. We only measure the hit rates for the read operations and do not consider writes since we are only interested in memory-read operations resulting from traversing the graph. In Figure~\ref{fig:cache_hitrates}, we see that BOBA achieves similar cache hit rates to other heavyweight techniques (i.e., Gorder and RCM) for the TC, SpMV, and  PR algorithms. Other lightweight reordering techniques achieve cache hit rates closer to random than heavyweight reorderings.

TC has high data reuse; hence, it enjoys a very high hit rate (specifically, at the L1 cache level). SSSP shows the least improvement from reordering. In general, BOBA achieves the cache performance of heavyweight methods at the reordering cost of a lightweight method.

\begin{figure}
    \centering
    \captionsetup{font=footnotesize,labelfont=footnotesize}
    \includegraphics[width=\columnwidth]{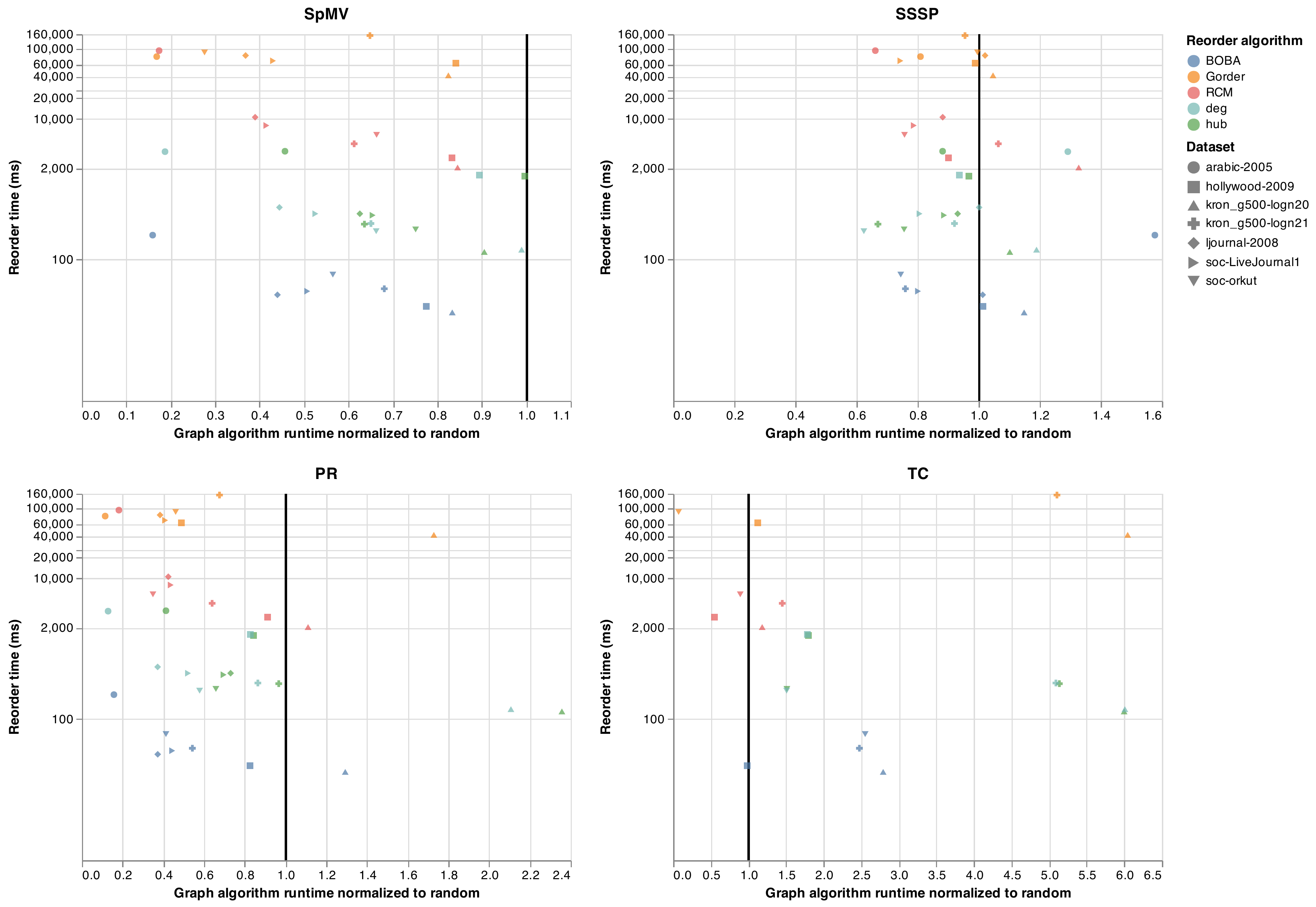}
    \caption{Runtime vs.\ reorder time for scale-free graphs.}
    \label{fig:scale_free_pareto}
\end{figure}

\begin{figure}
    \centering
    \captionsetup{font=footnotesize,labelfont=footnotesize}
    \includegraphics[width=\columnwidth]{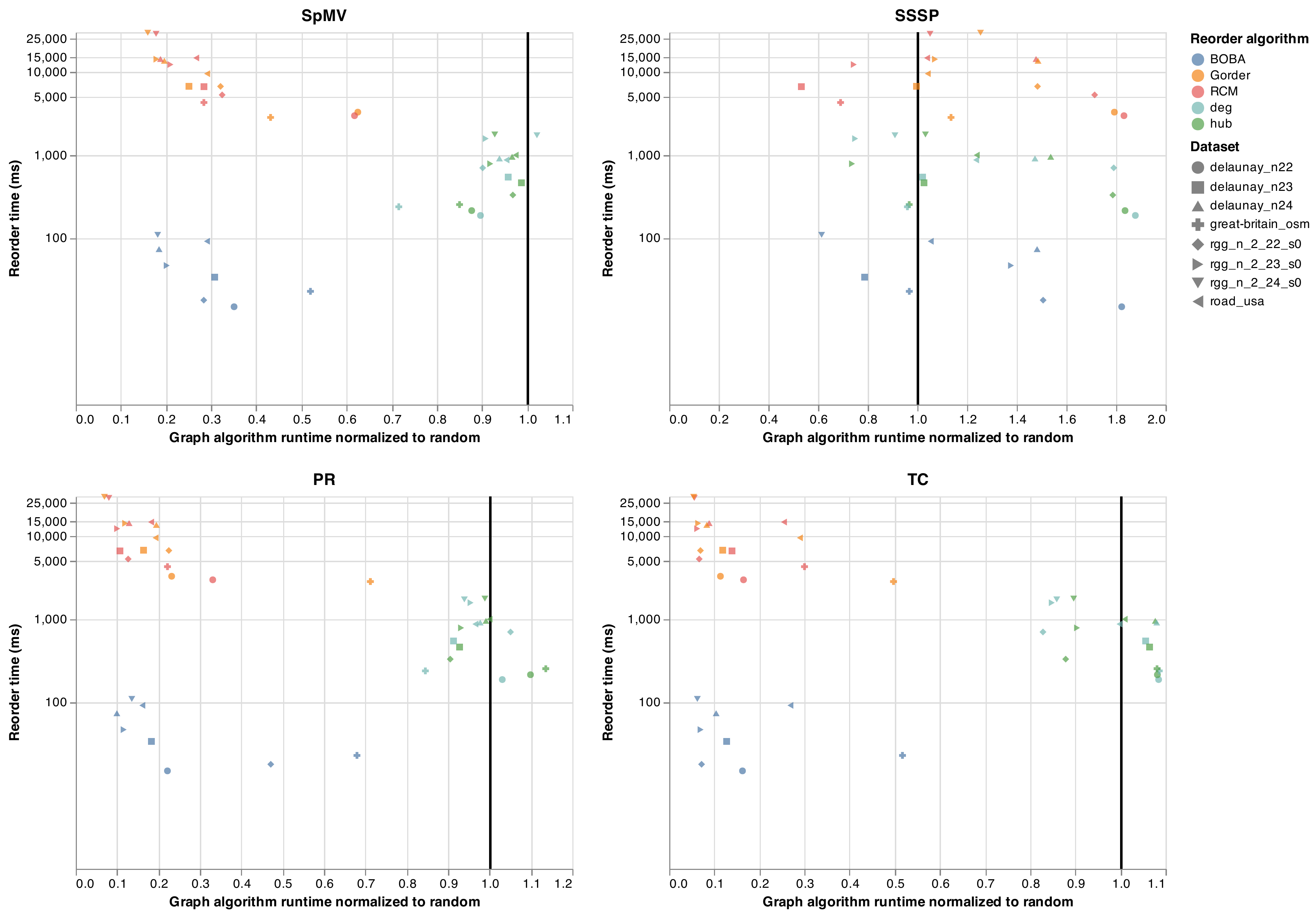}
    \caption{Runtime vs.\ reorder time for uniform graphs. Degree-based techniques perform worse than random on road networks where degree is anticorrelated to connectivity (Also see Figure~\ref{fig:road}). BOBA-ordered graphs achieve similar performance to heaveweight reorderings.}
    \label{fig:uniform_pareto}
\end{figure}

\begin{figure}
    \centering

    \captionsetup{font=footnotesize,labelfont=footnotesize}
    \begin{subfigure}[t]{\columnwidth}
        \includegraphics[width=\columnwidth]{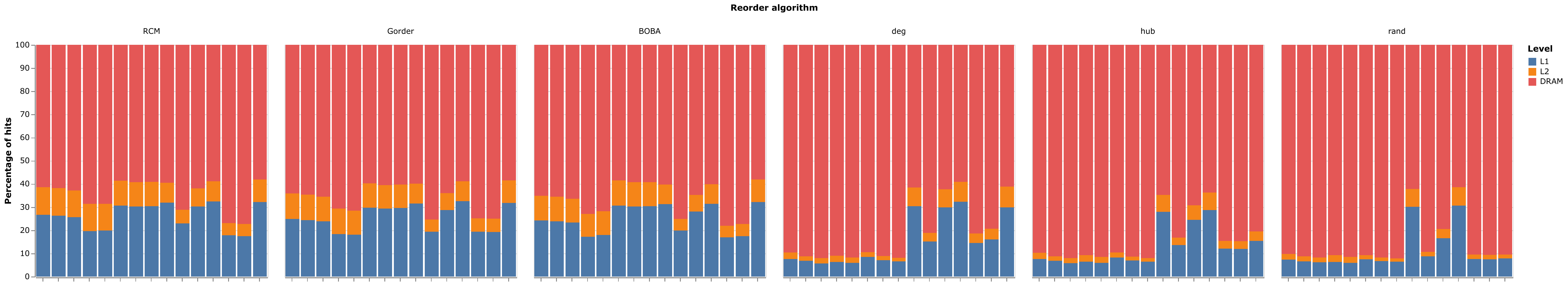}
        \includegraphics[width=\columnwidth]{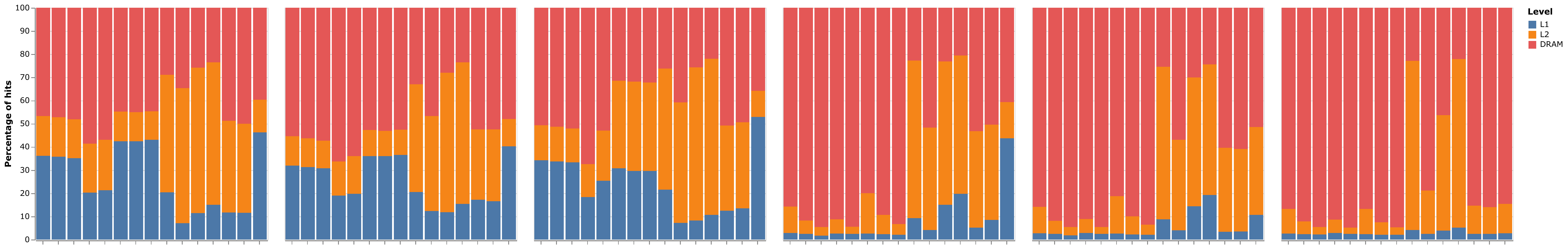}
        \includegraphics[width=\columnwidth]{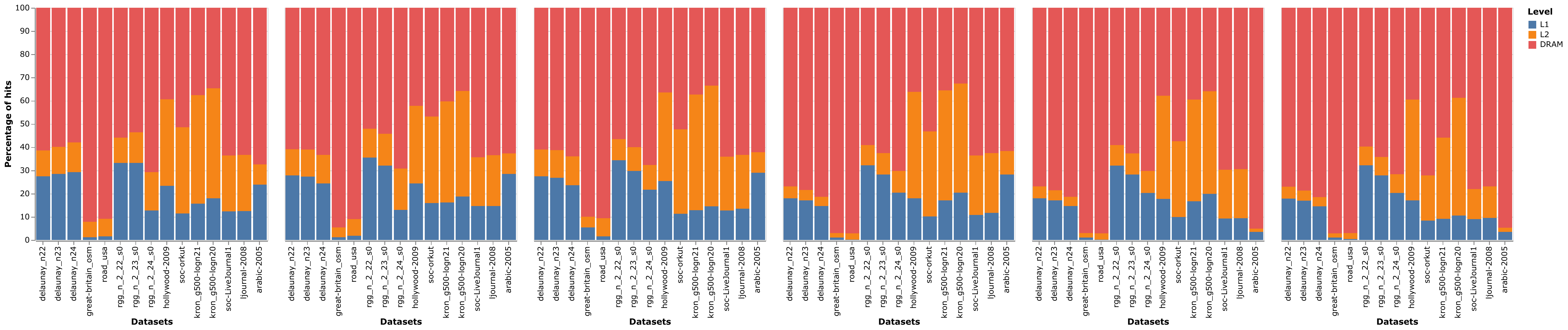}
        \caption{PR (top), SpMV (middle), and SSSP (bottom)}
    \end{subfigure}

    \begin{subfigure}[t]{\columnwidth}
        \includegraphics[width=\columnwidth]{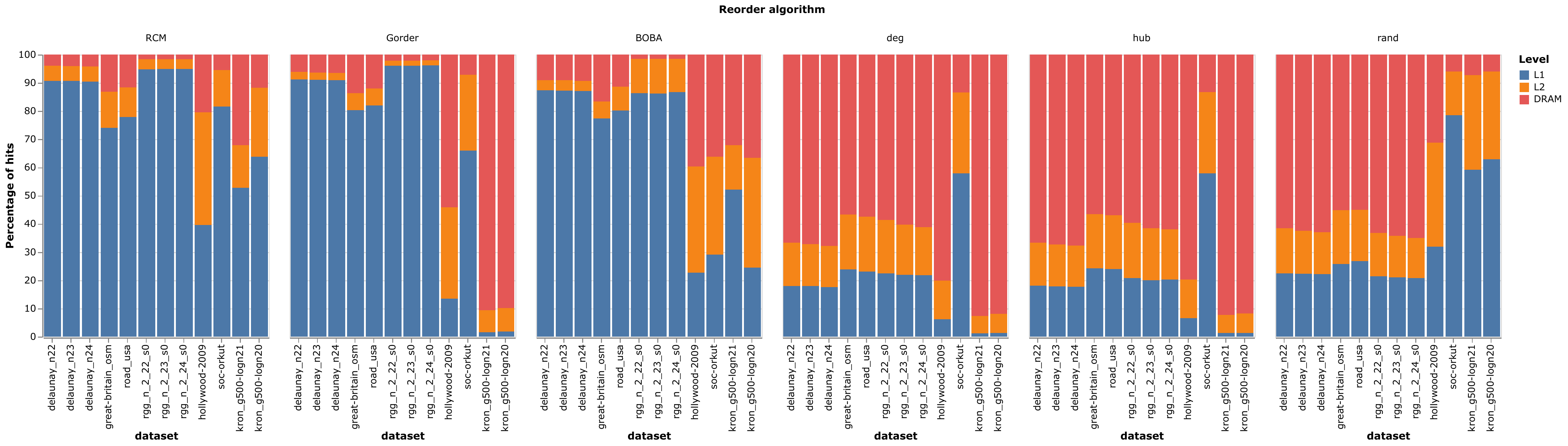}
        \caption{TC}
    \end{subfigure}

    \caption{Analysis of cache hit rates for different algorithms. BOBA achieves roughly the same cache hit rates as the best other methods with a much smaller reordering cost.}
    \label{fig:cache_hitrates}
\end{figure}

\subsection{Randomized Edge Orders}
\label{subsec:randomEdges}
Of course BOBA is sensitive to the order of the input edges in the input COO matrix. It is common to assume that the input edge list is sorted by either destination or source~\cite{Filippone:2017:SMM}. The data sets and repositories we investigated were no exception to this, but obviously the edges, especially of road-type networks without strong degree distributions, can be adversarially ordered such that BOBA cannot help. If an edge list does appear in random order, we suggest either sorting or binning the COO by destination, if the cost is acceptable, before running BOBA\@.

Table~\ref{tab:random} lists the  COO to CSR conversion and \spmv run times after applying BOBA to several datasets that were randomized before. BOBA provides no advantage on a more uniform network (delaunay\_n22) but shows modest performance gains as the networks become more scale-free.

\begin{table}
    \setlength{\belowcaptionskip}{0in}
    \setlength{\abovecaptionskip}{0in}

    \captionsetup{font=footnotesize,labelfont=footnotesize}
    \begin{small}
        \begin{tabular}{lccccc}
            \toprule
            \multirow{2}{*}{Randomized datasets} & \multicolumn{2}{c}{Rand} & \multicolumn{2}{c}{BOBA}                \\                               & \multicolumn{1}{l}{\spmv}&\multicolumn{1}{l}{$COO\rightarrow CSR$}\\
            \midrule
            arabic-2005                          & 97                       & 118207                   & 42  & 109351 \\
            soc-LiveJournal                      & 12.6                     & 8920                     & 7.5 & 8268   \\
            delaunay\_n22                        & 5.4                      & 3071                     & 5.5 & 3105   \\
            coPapersCiteseer                     & 3.4                      & 3239                     & 3.4 & 2891   \\ \bottomrule
        \end{tabular}
    \end{small}
    \caption{We first randomized four datasets and then ran BOBA on the randomized datasets. The runtimes above (in ms) compare \spmv and COO$\rightarrow$CSR conversion between the randomized dataset and the BOBA-reordered randomized dataset.}
    \label{tab:random}
\end{table}

\section{Conclusions and Future Work}
\label{sec:conclusion}

We introduced BOBA, a lightweight parallel-friendly fast reordering algorithm that achieves heavyweight-like improvements in cache locality and memory access patterns. We believe that BOBA can be easily extended beyond our single GPU implementation. BOBA will scale well with more GPUs, and the increased cache locality delivered by a BOBA reordering will hopefully translate to a multi-GPU setting. With emerging unified memory support  (i.e., CUDA's unified memory) across multiple GPUs, we believe that BOBA can reduce inter-GPU communication volume on multi-GPU graph primitives.

Moreover, in the future, we would like to address questions such as: Are there real-time or streaming applications for BOBA?  Can BOBA, or a BOBA-like method, give similar speed-ups to other pragmatic workflows over lists of structures, such as tensors, that can be modeled as hypergraphs? We leave these as open questions.

\bibliography{references, ../bib/all}

%%% -*-BibTeX-*-
%%% Do NOT edit. File created by BibTeX with style
%%% ACM-Reference-Format-Journals [18-Jan-2012].

\begin{thebibliography}{29}

%%% ====================================================================
%%% NOTE TO THE USER: you can override these defaults by providing
%%% customized versions of any of these macros before the \bibliography
%%% command.  Each of them MUST provide its own final punctuation,
%%% except for \shownote{}, \showDOI{}, and \showURL{}.  The latter two
%%% do not use final punctuation, in order to avoid confusing it with
%%% the Web address.
%%%
%%% To suppress output of a particular field, define its macro to expand
%%% to an empty string, or better, \unskip, like this:
%%%
%%% \newcommand{\showDOI}[1]{\unskip}   % LaTeX syntax
%%%
%%% \def \showDOI #1{\unskip}           % plain TeX syntax
%%%
%%% ====================================================================

\ifx \showCODEN    \undefined \def \showCODEN     #1{\unskip}     \fi
\ifx \showDOI      \undefined \def \showDOI       #1{#1}\fi
\ifx \showISBNx    \undefined \def \showISBNx     #1{\unskip}     \fi
\ifx \showISBNxiii \undefined \def \showISBNxiii  #1{\unskip}     \fi
\ifx \showISSN     \undefined \def \showISSN      #1{\unskip}     \fi
\ifx \showLCCN     \undefined \def \showLCCN      #1{\unskip}     \fi
\ifx \shownote     \undefined \def \shownote      #1{#1}          \fi
\ifx \showarticletitle \undefined \def \showarticletitle #1{#1}   \fi
\ifx \showURL      \undefined \def \showURL       {\relax}        \fi
% The following commands are used for tagged output and should be
% invisible to TeX
\providecommand\bibfield[2]{#2}
\providecommand\bibinfo[2]{#2}
\providecommand\natexlab[1]{#1}
\providecommand\showeprint[2][]{arXiv:#2}

\bibitem[Albert and Barab{\'a}si(2002)]%
        {albert2002statistical}
\bibfield{author}{\bibinfo{person}{R{\'e}ka Albert} {and}
  \bibinfo{person}{Albert-L{\'a}szl{\'o} Barab{\'a}si}.}
  \bibinfo{year}{2002}\natexlab{}.
\newblock \showarticletitle{Statistical mechanics of complex networks}.
\newblock \bibinfo{journal}{\emph{Reviews of Modern Physics}}
  \bibinfo{volume}{74}, \bibinfo{number}{1} (\bibinfo{date}{Jan.}
  \bibinfo{year}{2002}), \bibinfo{pages}{47--97}.
\newblock
\showISSN{1539-0756}
\urldef\tempurl%
\url{https://doi.org/10.1103/revmodphys.74.47}
\showDOI{\tempurl}


\bibitem[Amestoy et~al\mbox{.}(2004)]%
        {amestoy1996approximate}
\bibfield{author}{\bibinfo{person}{Patrick~R. Amestoy},
  \bibinfo{person}{Timothy~A. Davis}, {and} \bibinfo{person}{Iain~S. Duff}.}
  \bibinfo{year}{2004}\natexlab{}.
\newblock \showarticletitle{Algorithm 837: {AMD}, An Approximate Minimum Degree
  Ordering Algorithm}.
\newblock \bibinfo{journal}{\emph{ACM Trans. Math. Software}}
  \bibinfo{volume}{30}, \bibinfo{number}{3} (\bibinfo{date}{Sept.}
  \bibinfo{year}{2004}), \bibinfo{pages}{381--388}.
\newblock
\showISSN{1557-7295}
\urldef\tempurl%
\url{https://doi.org/10.1145/1024074.1024081}
\showDOI{\tempurl}


\bibitem[Anderson et~al\mbox{.}(1999)]%
        {anderson:1999:achieving}
\bibfield{author}{\bibinfo{person}{W.~K. Anderson}, \bibinfo{person}{W.~D.
  Gropp}, \bibinfo{person}{D.~K. Kaushik}, \bibinfo{person}{D.~E. Keyes}, {and}
  \bibinfo{person}{B.~F. Smith}.} \bibinfo{year}{1999}\natexlab{}.
\newblock \showarticletitle{Achieving High Sustained Performance in an
  Unstructured Mesh {CFD} Application}.
\newblock \bibinfo{journal}{\emph{Proceedings of the 1999 ACM/IEEE Conference
  on Supercomputing}}.
\newblock
\urldef\tempurl%
\url{https://doi.org/10.1145/331532.331600}
\showDOI{\tempurl}


\bibitem[Arai et~al\mbox{.}(2016)]%
        {arai2016rabbit}
\bibfield{author}{\bibinfo{person}{Junya Arai}, \bibinfo{person}{Hiroaki
  Shiokawa}, \bibinfo{person}{Takeshi Yamamuro}, \bibinfo{person}{Makoto
  Onizuka}, {and} \bibinfo{person}{Sotetsu Iwamura}.}
  \bibinfo{year}{2016}\natexlab{}.
\newblock \showarticletitle{Rabbit order: Just-in-time parallel reordering for
  fast graph analysis}. In \bibinfo{booktitle}{\emph{2016 IEEE International
  Parallel and Distributed Processing Symposium (IPDPS)}}. IEEE,
  \bibinfo{pages}{22--31}.
\newblock


\bibitem[Awad et~al\mbox{.}(2021)]%
        {Awad:2021:BGH}
\bibfield{author}{\bibinfo{person}{Muhammad~A. Awad}, \bibinfo{person}{Saman
  Ashkiani}, \bibinfo{person}{Serban~D. Porumbescu},
  \bibinfo{person}{Mart{\'{i}}n Farach-Colton}, {and} \bibinfo{person}{John~D.
  Owens}.} \bibinfo{year}{2021}\natexlab{}.
\newblock \showarticletitle{Better {GPU} Hash Tables}.
\newblock \bibinfo{journal}{\emph{CoRR}} \bibinfo{volume}{abs/2108.07232},
  \bibinfo{number}{2108.07232} (\bibinfo{date}{Aug.} \bibinfo{year}{2021}).
\newblock
\showeprint[arxiv]{2108.07232}~[cs.DS]


\bibitem[Balaji and Lucia(2018)]%
        {balaji2018graph}
\bibfield{author}{\bibinfo{person}{Vignesh Balaji} {and}
  \bibinfo{person}{Brandon Lucia}.} \bibinfo{year}{2018}\natexlab{}.
\newblock \showarticletitle{When is Graph Reordering an Optimization? Studying
  the Effect of Lightweight Graph Reordering Across Applications and Input
  Graphs}. In \bibinfo{booktitle}{\emph{2018 IEEE International Symposium on
  Workload Characterization}} \emph{(\bibinfo{series}{IISWC 2018})}.
  \bibinfo{publisher}{IEEE}, \bibinfo{pages}{203--214}.
\newblock
\urldef\tempurl%
\url{https://doi.org/10.1109/iiswc.2018.8573478}
\showDOI{\tempurl}


\bibitem[Bollob\'{a}s and Riordan(2004a)]%
        {Bollobas:2004:TDO}
\bibfield{author}{\bibinfo{person}{B\'{e}la Bollob\'{a}s} {and}
  \bibinfo{person}{Oliver Riordan}.} \bibinfo{year}{2004}\natexlab{a}.
\newblock \showarticletitle{The Diameter of a Scale-Free Random Graph}.
\newblock \bibinfo{journal}{\emph{Combinatorica}} \bibinfo{volume}{24},
  \bibinfo{number}{1} (\bibinfo{date}{Jan.} \bibinfo{year}{2004}),
  \bibinfo{pages}{5--34}.
\newblock
\showISSN{0209-9683}
\urldef\tempurl%
\url{https://doi.org/10.1007/s00493-004-0002-2}
\showDOI{\tempurl}


\bibitem[Bollob\'{a}s and Riordan(2004b)]%
        {Bollobas:2002:MRO}
\bibfield{author}{\bibinfo{person}{B\'{e}la Bollob\'{a}s} {and}
  \bibinfo{person}{Oliver~M. Riordan}.} \bibinfo{year}{2004}\natexlab{b}.
\newblock \showarticletitle{Mathematical results on scale-free random graphs}.
\newblock In \bibinfo{booktitle}{\emph{Handbook of Graphs and Networks}}.
  Chapter~1, \bibinfo{pages}{1--34}.
\newblock
\showISBNx{9783527602759}
\urldef\tempurl%
\url{https://doi.org/10.1002/3527602755.ch1}
\showDOI{\tempurl}


\bibitem[Chen and Chung(2022)]%
        {Chen:2022:WBG}
\bibfield{author}{\bibinfo{person}{YuAng Chen} {and} \bibinfo{person}{Yeh-Ching
  Chung}.} \bibinfo{year}{2022}\natexlab{}.
\newblock \showarticletitle{Workload Balancing via Graph Reordering on
  Multicore Systems}.
\newblock \bibinfo{journal}{\emph{IEEE Transactions on Parallel and Distributed
  Systems}} \bibinfo{volume}{33}, \bibinfo{number}{5} (\bibinfo{year}{2022}),
  \bibinfo{pages}{1231--1245}.
\newblock
\urldef\tempurl%
\url{https://doi.org/10.1109/TPDS.2021.3105323}
\showDOI{\tempurl}


\bibitem[Cuthill and McKee(1969)]%
        {Cuthill:1969:RTB}
\bibfield{author}{\bibinfo{person}{E. Cuthill} {and} \bibinfo{person}{J.
  McKee}.} \bibinfo{year}{1969}\natexlab{}.
\newblock \showarticletitle{Reducing the Bandwidth of Sparse Symmetric
  Matrices}. In \bibinfo{booktitle}{\emph{Proceedings of the 1969 24th National
  Conference}} \emph{(\bibinfo{series}{ACM 1969})}. \bibinfo{pages}{157--172}.
\newblock
\showISBNx{9781450374934}
\urldef\tempurl%
\url{https://doi.org/10.1145/800195.805928}
\showDOI{\tempurl}


\bibitem[Davis and Hu(2011)]%
        {Davis:2011:Florida}
\bibfield{author}{\bibinfo{person}{Timothy~A. Davis} {and}
  \bibinfo{person}{Yifan Hu}.} \bibinfo{year}{2011}\natexlab{}.
\newblock \showarticletitle{The University of Florida Sparse Matrix
  Collection}.
\newblock \bibinfo{journal}{\emph{ACM Transactions on Mathematical Software 38,
  1, Article 1}} \bibinfo{volume}{8}, \bibinfo{number}{1}
  (\bibinfo{year}{2011}), \bibinfo{pages}{25}.
\newblock
\urldef\tempurl%
\url{https://doi.org/10.1145/2049662.2049663}
\showDOI{\tempurl}


\bibitem[Esfahani et~al\mbox{.}(2021)]%
        {esfahani:2021:locality}
\bibfield{author}{\bibinfo{person}{Mohsen~Koohi Esfahani},
  \bibinfo{person}{Peter Kilpatrick}, {and} \bibinfo{person}{Hans
  Vandierendonck}.} \bibinfo{year}{2021}\natexlab{}.
\newblock \showarticletitle{Locality Analysis of Graph Reordering Algorithms}.
  In \bibinfo{booktitle}{\emph{2021 IEEE International Symposium on Workload
  Characterization (IISWC)}}. IEEE, \bibinfo{pages}{101--112}.
\newblock
\urldef\tempurl%
\url{https://doi.org/10.1109/IISWC53511.2021.00020}
\showDOI{\tempurl}


\bibitem[Eubank et~al\mbox{.}(2004)]%
        {eubank2004structural}
\bibfield{author}{\bibinfo{person}{Stephen Eubank}, \bibinfo{person}{V.~S.~Anil
  Kumar}, \bibinfo{person}{Madhav~V. Marathe}, \bibinfo{person}{Aravind
  Srinivasan}, {and} \bibinfo{person}{Nan Wang}.}
  \bibinfo{year}{2004}\natexlab{}.
\newblock \showarticletitle{Structural and algorithmic aspects of massive
  social networks}. In \bibinfo{booktitle}{\emph{Proceedings of the Fifteenth
  Annual ACM-SIAM Symposium on Discrete Algorithms}}
  \emph{(\bibinfo{series}{SODA '04})}. \bibinfo{pages}{718--727}.
\newblock


\bibitem[Faldu et~al\mbox{.}(2019)]%
        {faldu2019closer}
\bibfield{author}{\bibinfo{person}{Priyank Faldu}, \bibinfo{person}{Jeff
  Diamond}, {and} \bibinfo{person}{Boris Grot}.}
  \bibinfo{year}{2019}\natexlab{}.
\newblock \showarticletitle{A closer look at lightweight graph reordering}. In
  \bibinfo{booktitle}{\emph{2019 IEEE International Symposium on Workload
  Characterization (IISWC)}}. IEEE, \bibinfo{pages}{1--13}.
\newblock


\bibitem[Filippone et~al\mbox{.}(2017)]%
        {Filippone:2017:SMM}
\bibfield{author}{\bibinfo{person}{Salvatore Filippone},
  \bibinfo{person}{Valeria Cardellini}, \bibinfo{person}{Davide Barbieri},
  {and} \bibinfo{person}{Alessandro Fanfarillo}.}
  \bibinfo{year}{2017}\natexlab{}.
\newblock \showarticletitle{Sparse Matrix-Vector Multiplication on {GPGPU}s}.
\newblock \bibinfo{journal}{\emph{ACM Trans. Math. Softw.}}
  \bibinfo{volume}{43}, \bibinfo{number}{4}, Article \bibinfo{articleno}{30}
  (\bibinfo{date}{Jan.} \bibinfo{year}{2017}), \bibinfo{numpages}{49}~pages.
\newblock
\showISSN{0098-3500}
\urldef\tempurl%
\url{https://doi.org/10.1145/3017994}
\showDOI{\tempurl}


\bibitem[Green et~al\mbox{.}(2012)]%
        {Green:2012:GMP}
\bibfield{author}{\bibinfo{person}{Oded Green}, \bibinfo{person}{Robert
  McColl}, {and} \bibinfo{person}{David~A. Bader}.}
  \bibinfo{year}{2012}\natexlab{}.
\newblock \showarticletitle{{GPU} Merge Path: A {GPU} Merging Algorithm}. In
  \bibinfo{booktitle}{\emph{Proceedings of the 26th ACM International
  Conference on Supercomputing}} (San Servolo Island, Venice, Italy)
  \emph{(\bibinfo{series}{ICS '12})}. \bibinfo{pages}{331--340}.
\newblock
\showISBNx{978-1-4503-1316-2}
\urldef\tempurl%
\url{https://doi.org/10.1145/2304576.2304621}
\showDOI{\tempurl}


\bibitem[Karantasis et~al\mbox{.}(2014)]%
        {Karantasis:2014:PRA}
\bibfield{author}{\bibinfo{person}{Konstantinos~I. Karantasis},
  \bibinfo{person}{Andrew Lenharth}, \bibinfo{person}{Donald Nguyen},
  \bibinfo{person}{Mara~J. Garzaran}, {and} \bibinfo{person}{Keshav Pingali}.}
  \bibinfo{year}{2014}\natexlab{}.
\newblock \showarticletitle{Parallelization of Reordering Algorithms for
  Bandwidth and Wavefront Reduction}. In \bibinfo{booktitle}{\emph{Proceedings
  of the International Conference for High Performance Computing, Networking,
  Storage and Analysis}} \emph{(\bibinfo{series}{SC '14})}.
  \bibinfo{publisher}{IEEE}, \bibinfo{pages}{921--932}.
\newblock
\urldef\tempurl%
\url{https://doi.org/10.1109/sc.2014.80}
\showDOI{\tempurl}


\bibitem[Leskovec and Sosi{\v{c}}(2016)]%
        {leskovec:2016:snap}
\bibfield{author}{\bibinfo{person}{Jure Leskovec} {and} \bibinfo{person}{Rok
  Sosi{\v{c}}}.} \bibinfo{year}{2016}\natexlab{}.
\newblock \showarticletitle{{SNAP}: A General-Purpose Network Analysis and
  Graph-Mining Library}.
\newblock \bibinfo{journal}{\emph{ACM Transactions on Intelligent Systems and
  Technology}} \bibinfo{volume}{8}, \bibinfo{number}{1} (\bibinfo{date}{Oct.}
  \bibinfo{year}{2016}), \bibinfo{pages}{1--20}.
\newblock
\showISSN{2157-6912}
\urldef\tempurl%
\url{https://doi.org/10.1145/2898361}
\showDOI{\tempurl}


\bibitem[Liu(1976)]%
        {liu:1976:reducing}
\bibfield{author}{\bibinfo{person}{Joseph Wai-Hung Liu}.}
  \bibinfo{year}{1976}\natexlab{}.
\newblock \emph{\bibinfo{title}{On reducing the profile of sparse symmetric
  matrices}}.
\newblock \bibinfo{thesistype}{Ph.\,D. Dissertation}. \bibinfo{school}{Faculty
  of Mathematics, University of Waterloo}.
\newblock


\bibitem[Merrill and Garland(2016)]%
        {Merrill:2016:MPS}
\bibfield{author}{\bibinfo{person}{Duane Merrill} {and}
  \bibinfo{person}{Michael Garland}.} \bibinfo{year}{2016}\natexlab{}.
\newblock \showarticletitle{Merge-Based Parallel Sparse Matrix-Vector
  Multiplication}. In \bibinfo{booktitle}{\emph{International Conference for
  High Performance Computing, Networking, Storage and Analysis}}
  \emph{(\bibinfo{series}{SC '16})}. \bibinfo{pages}{678--689}.
\newblock
\showISBNx{9781467388153}
\urldef\tempurl%
\url{https://doi.org/10.1109/SC.2016.57}
\showDOI{\tempurl}


\bibitem[Newman(2018)]%
        {newman2018networks}
\bibfield{author}{\bibinfo{person}{Mark Newman}.}
  \bibinfo{year}{2018}\natexlab{}.
\newblock \bibinfo{booktitle}{\emph{Networks} (\bibinfo{edition}{second} ed.)}.
\newblock \bibinfo{publisher}{Oxford University Press}.
\newblock
\showISBNx{978-0198805090}


\bibitem[Osama et~al\mbox{.}(2022)]%
        {Osama:2022:EOP}
\bibfield{author}{\bibinfo{person}{Muhammad Osama}, \bibinfo{person}{Serban~D.
  Porumbescu}, {and} \bibinfo{person}{John~D. Owens}.}
  \bibinfo{year}{2022}\natexlab{}.
\newblock \showarticletitle{Essentials of Parallel Graph Analytics}. In
  \bibinfo{booktitle}{\emph{Proceedings of the Workshop on Graphs,
  Architectures, Programming, and Learning}} \emph{(\bibinfo{series}{GrAPL
  2022})}.
\newblock
\urldef\tempurl%
\url{https://escholarship.org/uc/item/2p19z28q}
\showURL{%
\tempurl}


\bibitem[Papadimitriou(1976)]%
        {papadimitriou1976np}
\bibfield{author}{\bibinfo{person}{Ch.~H. Papadimitriou}.}
  \bibinfo{year}{1976}\natexlab{}.
\newblock \showarticletitle{The {NP}-{C}ompleteness of the bandwidth
  minimization problem}.
\newblock \bibinfo{journal}{\emph{Computing}} \bibinfo{volume}{16},
  \bibinfo{number}{3} (\bibinfo{date}{Sept.} \bibinfo{year}{1976}),
  \bibinfo{pages}{263--270}.
\newblock
\showISSN{1436-5057}
\urldef\tempurl%
\url{https://doi.org/10.1007/bf02280884}
\showDOI{\tempurl}


\bibitem[Sanders et~al\mbox{.}(2019)]%
        {sanders2019sequential}
\bibfield{author}{\bibinfo{person}{Peter Sanders}, \bibinfo{person}{Kurt
  Mehlhorn}, \bibinfo{person}{Martin Dietzfelbinger}, {and}
  \bibinfo{person}{Roman Dementiev}.} \bibinfo{year}{2019}\natexlab{}.
\newblock \bibinfo{booktitle}{\emph{Sequential and Parallel Algorithms and Data
  Structures}}.
\newblock \bibinfo{publisher}{Springer}.
\newblock


\bibitem[Sloan(1986)]%
        {sloan:1986:algorithm}
\bibfield{author}{\bibinfo{person}{S.~W. Sloan}.}
  \bibinfo{year}{1986}\natexlab{}.
\newblock \showarticletitle{An algorithm for profile and wavefront reduction of
  sparse matrices}.
\newblock \bibinfo{journal}{\emph{Internat. J. Numer. Methods Engrg.}}
  \bibinfo{volume}{23}, \bibinfo{number}{2} (\bibinfo{date}{Feb.}
  \bibinfo{year}{1986}), \bibinfo{pages}{239--251}.
\newblock
\showISSN{1097-0207}
\urldef\tempurl%
\url{https://doi.org/10.1002/nme.1620230208}
\showDOI{\tempurl}


\bibitem[The Mathworks, Inc.(2021)]%
        {MATLAB:R2021a}
The Mathworks, Inc. \bibinfo{year}{2021}\natexlab{}.
\newblock \bibinfo{booktitle}{\emph{{MATLAB version 9.10.0.1613233 (R2021a)}}}.
\newblock The Mathworks, Inc., Natick, Massachusetts.
\newblock


\bibitem[Wang et~al\mbox{.}(2017)]%
        {Wang:2017:GGG}
\bibfield{author}{\bibinfo{person}{Yangzihao Wang}, \bibinfo{person}{Yuechao
  Pan}, \bibinfo{person}{Andrew Davidson}, \bibinfo{person}{Yuduo Wu},
  \bibinfo{person}{Carl Yang}, \bibinfo{person}{Leyuan Wang},
  \bibinfo{person}{Muhammad Osama}, \bibinfo{person}{Chenshan Yuan},
  \bibinfo{person}{Weitang Liu}, \bibinfo{person}{Andy~T. Riffel}, {and}
  \bibinfo{person}{John~D. Owens}.} \bibinfo{year}{2017}\natexlab{}.
\newblock \showarticletitle{{G}unrock: {GPU} Graph Analytics}.
\newblock \bibinfo{journal}{\emph{ACM Transactions on Parallel Computing}}
  \bibinfo{volume}{4}, \bibinfo{number}{1} (\bibinfo{date}{Aug.}
  \bibinfo{year}{2017}), \bibinfo{pages}{3:1--3:49}.
\newblock
\urldef\tempurl%
\url{https://doi.org/10.1145/3108140}
\showDOI{\tempurl}


\bibitem[Wei et~al\mbox{.}(2016)]%
        {Wei:2016:SGP}
\bibfield{author}{\bibinfo{person}{Hao Wei}, \bibinfo{person}{Jeffrey~Xu Yu},
  \bibinfo{person}{Can Lu}, {and} \bibinfo{person}{Xuemin Lin}.}
  \bibinfo{year}{2016}\natexlab{}.
\newblock \showarticletitle{Speedup Graph Processing by Graph Ordering}. In
  \bibinfo{booktitle}{\emph{Proceedings of the 2016 International Conference on
  Management of Data}} \emph{(\bibinfo{series}{SIGMOD 2016})}.
  \bibinfo{pages}{1813--1828}.
\newblock
\showISBNx{9781450335317}
\urldef\tempurl%
\url{https://doi.org/10.1145/2882903.2915220}
\showDOI{\tempurl}


\bibitem[Zhang et~al\mbox{.}(2017)]%
        {zhang2017making}
\bibfield{author}{\bibinfo{person}{Yunming Zhang}, \bibinfo{person}{Vladimir
  Kiriansky}, \bibinfo{person}{Charith Mendis}, \bibinfo{person}{Saman
  Amarasinghe}, {and} \bibinfo{person}{Matei Zaharia}.}
  \bibinfo{year}{2017}\natexlab{}.
\newblock \showarticletitle{Making Caches Work for Graph Analytics}. In
  \bibinfo{booktitle}{\emph{2017 IEEE International Conference on Big Data}}
  \emph{(\bibinfo{series}{BigData 2017})}. \bibinfo{publisher}{IEEE},
  \bibinfo{pages}{293--302}.
\newblock
\urldef\tempurl%
\url{https://doi.org/10.1109/bigdata.2017.8257937}
\showDOI{\tempurl}


\end{thebibliography}

\end{document}